\documentclass[11pt,final]{article}

\usepackage{timesjhep}
\usepackage{slashed,fancybox,amsmath}
\usepackage[capitalize]{cleveref}
\crefformat{equation}{(#2#1#3)}
\crefrangeformat{equation}{(#3#1#4--#5#2#6)}
\usepackage[nomargin,inline]{fixme}
\makeatletter
\renewcommand*\FXLayoutInline[3]{{\@fxuseface{inline}
  \ignorespaces\noindent \ovalbox{\hspace{.01\textwidth} \begin{minipage}{.95\textwidth}
  	#3 \fxnotename{#1}: #2
  \end{minipage}\hspace{.01\textwidth}}}
  \newline}
\makeatother
\newcommand{\Ben}[1]{\fxnote[inline,author=Ben]{\textcolor{blue!50}{ #1}}}

\usepackage{verbatim}
\usepackage[T1]{fontenc}
\usepackage[utf8]{inputenc}
\usepackage{svg}
\usepackage{lipsum}
\usepackage[english]{babel}
\usepackage{physics}
\usepackage{dcolumn}
\usepackage{tensor}
\usepackage{comment}
\usepackage{overpic,mathtools}
\usepackage{braket,bm,bbm,setspace}
\usepackage{cancel}
\usepackage{float}
\usepackage{simpler-wick}
\usepackage{tikz,tikz-cd,nicefrac}
\allowdisplaybreaks
\theoremstyle{definition}
\newtheorem{theorem}{Theorem}[section]
\newtheorem{cor}[theorem]{Corollary}
\newtheorem{lemma}[theorem]{Lemma}
\newtheorem{prop}[theorem]{Proposition}
\newtheorem{defi}[theorem]{Definition}

\newtheorem{example}[theorem]{Example}
\newtheorem{remark}[theorem]{Remark}

\newcommand{\Weyl}{\mathscr{W}}
\newcommand{\LE}{\beta}
\newcommand{\C}{\mathbb{C}}
\newcommand{\fg}{\mathfrak{g}}
\newcommand{\Z}{\mathbb{Z}}
\newcommand{\Van}{\mathbb{V}}
\newcommand{\No}{H}
\newcommand{\no}{\mathfrak{h}}
\newcommand{\ft}{\mathfrak{t}}
\newcommand{\To}{Q}
\newcommand{\mto}{\mathfrak{q}}
\newcommand{\cO}{\mathcal{O}}
\newcommand{\red}{\mathfrak{r}}
\newcommand{\mmod}{\operatorname{-mod}}
\newcommand{\mass}{\mathbf m}
\newcommand{\FI}{\boldsymbol{\zeta}}
\newcommand{\QI}{Q_I}
\newcommand{\qI}{R_I}
\newcommand{\qi}{r}
\newcommand{\Verma}{\Delta}
\newcommand{\onef}{(-1)^F}
\newcommand{\sorc}{\operatorname{CS}}
\newcommand{\antilinear}{\rho}
\newcommand{\signs}{\boldsymbol{\sigma}}
\renewcommand{\dh}{\mathbf{h}}
\newcommand{\colsum}{\mathsf{r}}
\newcommand{\rowsum}{\mathsf{s}}

\newcommand{\AHiggs}{A_{\operatorname{Higgs}}}
\newcommand{\ACoulomb}{A_{\operatorname{Coulomb}}}
\newcommand{\AcHiggs}{A^{\operatorname{cl}}_{\operatorname{Higgs}}}
\newcommand{\AcCoulomb}{A^{\operatorname{cl}}_{\operatorname{Coulomb}}}

\newcommand{\weird}{{\wp}}
\newcommand{\roots}{R^+}

\title{Twisted traces on abelian quantum Higgs and Coulomb branches}

\abstract{We study twisted traces on the quantum Higgs branches $\AHiggs$ of $3d, \mathcal{N}=4$ gauge theories, that is, the quantum Hamiltonian reductions of Weyl algebras.  In theories which are good or ugly, we define a twisted trace that arises naturally from the correlation functions of the gauge theory.  We show that this trace induces an inner product and a short star product on $\AHiggs$.  

We analyze this trace in the case of an abelian gauge group and show that it has a natural expansion in terms of the twisted traces of Verma modules, confirming a conjecture of the first author and Okazaki.  This expansion has a natural interpretation in terms of 3-d mirror symmetry, and we predict that it can be interpreted as an Atiyah-Bott fixed-point formula under the quantum Hikita isomorphism.}

\begin{document}
\author[1]{Davide Gaiotto}
\emailAdd{dgaiotto@perimeterinstitute.ca}
\author[1]{Justin Hilburn}
\emailAdd{jhilburn@perimeterinstitute.ca}
\author[1,2]{Jaime Redondo-Yuste}
\emailAdd{jaime.redondo.yuste@nbi.ku.dk}
\author[1,3]{Ben Webster}
\emailAdd{bwebster@perimeterinstitute.ca}
\author[1,4]{Zheng Zhou}
\emailAdd{zzhou@perimeterinstitute.ca}
\affiliation[1]{
Perimeter Institute for Theoretical Physics, 31 Caroline Street North, Waterloo, ON N2L 2Y5, Canada
}
\affiliation[2]{Niels Bohr International Academy, Niels Bohr Institute, Blegdamsvej 17, 2100 Copenhagen, Denmark}{}
\affiliation[3]{Department of Pure Mathematics, University of Waterloo, Waterloo, Ontario, N2L 3G1,
Canada}
\affiliation[4]{Department of Physics and Astronomy, University of Waterloo, Waterloo, Ontario, N2L 3G1,
Canada}
\maketitle

\section{Introduction}

Symplectic singularities have become, over the last two decades, central objects of study in geometric representation theory.  While very rich and interesting objects in their own right, most of the interesting examples arise in quantum field theory, and studying them from this perspective gives us new and interesting insights.  

In particular, if we have a 3-dimensional  $\mathcal{N}=4$ SUSY field theory, then this theory has 2 topological twists, usually called the A- and B-twists.  Local operators in either of these theories naturally form an $\mathbb{E}_3$-algebra, and thus their cohomology is a graded commutative algebra with Poisson bracket of degree $-2$.   See \cite{BBBDN} for a more detailed discussion of this construction.  
We denote these algebras by $\AcCoulomb$ and $\AcHiggs$, respectively.

Taking the spectrum of these graded commutative algebras, we obtain two affine algebraic varieties, called the {\bf Coulomb branch} and {\bf Higgs branch} of the theory. These varieties also have natural noncommutative deformations (see \cite[\S 6]{BBBDN} for further discussion) to noncommutative algebras.  Examples include the Weyl algebra, universal enveloping algebras of $\mathfrak{sl}_n$ and rational Cherednik algebras, amongst others.  In this paper, we will only consider the case of gauge theories which arise from a symplectic representation $V$ of a compact group $G$.  In this case, the Higgs branch is a hyperk\"ahler reduction of $V$ by the group $G$ \cite[\S 6B]{hitchinHyperkahlerMetrics1987} and the Coulomb branch can be constructed as the spectrum of a convolution algebra rooted in the geometry of the affine Grassmannian of the group $G$ \cite{BFN}.  Both of these algebras have noncommutative deformations, which were considered by Chester, Lee, Pufu and Yacobi \cite{chesterMathcalSuperconformal2014} and Beem, Peelaers and Rastelli \cite{beemDeformationQuantization2017} in the context of the superconformal bootstrap in the IR limit of these theories; in \cite{beemDeformationQuantization2017}, these are called the {\bf protected associative algebra}.  In the case of the Higgs branch, this algebra is a noncommutative Hamiltonian reduction of a Weyl algebra; in the case of the Coulomb branch, this algebra is the equivariant homology of the same space with an additional $S^1$-action.  

In the usual language of TQFT, the commutative algebras $\AcCoulomb,\AcHiggs$ are the Hilbert spaces $Z(S^2)$ of the 2-sphere in these twisted theories. A reader used to TQFT will be a little concerned at this point, since a 3d TQFT in the usual sense should assign a finite-dimensional Frobenius algebra to $S^2$ with Frobenius trace given by the 3-ball, thought of as a cobordism;  in more physical terms, we would describe this as the correlation function of the 3-sphere with this operator inserted. The algebras $\AcCoulomb,\AcHiggs$ are not finite-dimensional in almost all cases, so it would appear that we have a problem.  However, the quantum field theory perspective suggests a replacement for the Frobenius property: These noncommutative deformations of the algebras possess a family of the ``special sphere correlation functions'' in \cite{gaiottoSphereCorrelation2020}, which we can interpret as a correlation function in the corresponding SCFT as discussed in \cite[\S 2.4]{beemDeformationQuantization2017}.

However, constructing these functions requires leaving the world of naive TQFT as mathematicians understand it.
Instead, what we naturally obtain is a correlation function $\Tr(\mathbf{a},\boldsymbol{\theta})$ that depends on an $n$-tuple $\mathbf{a}=(a_1,\dots, a_n)$ of elements of the algebra $\AHiggs$, and a cyclically oriented $n$-tuple of points on a great circle $S^1\subset S^3$, which we can identify with unit complex numbers $e^{2\pi i \theta_1},\dots, e^{2\pi i \theta_n}$.  In this context, it's more natural to endow $\AcHiggs$ with a star product deformation of its multiplication, that is, to replace it with the quantum Higgs branch algebra $\AHiggs$.  The ultimate correlation function we wish to understand is $\Tr(a)=\Tr((a),(0))$; we'll call this the {\bf sphere trace} of $\AHiggs$.  As discussed in \cite[\S 4]{dedushenkoOnedimensionalTheory2018}, the dependence on position is very simple in this case---if we vary $\theta_k$, we must vary $a_k$ by a 1-parameter family of elements of the centralizer of $G$ in $GL(N)$, thought of as inner automorphisms of $\AHiggs$ to compensate.  In physical terms, these are the exponential of multiples of the mass parameter $\mass$: 
\begin{equation}
    \Tr(\mathbf{a},\boldsymbol{\theta})=\Tr((a_1,\dots,e^{2\pi i t \mass}a_k,\dots, a_n),(\theta_1,\dots,\theta_k+t,\dots, \theta_n)).
\end{equation}
Thus, when we try to apply the usual argument that the trace of a product is invariant under cyclic permutation, we have to apply an automorphism corresponding to translating our operator around the circle, as well as an additional R-symmetry rotation $\onef$, which acts by an involution; in the notation of \cite{etingofShortStarProducts2020}, this map is denoted $s$.
\begin{multline}
    \Tr(ab)=\Tr((a,b),(0,0))=\Tr((a,e^{2\pi i  \mass}b),(0,2\pi))\\=\Tr((\onef e^{2\pi i  \mass}b,a),(0,0))=\Tr((\onef e^{2\pi i  \mass}b)a).    
\end{multline}
where now the product is taken in the non-commutative algebra $\AHiggs$.  
Thus, we obtain a
{\bf twisted trace}: 
\begin{defi}[Twisted trace]
    A twisted trace of an algebra $\mathscr{A}$ with respect to an automorphism $\omega$ is a linear map $\Tr_\omega: \mathscr{A}\rightarrow \mathbb{C}$ satisfying: 
    \begin{equation}
        \Tr_\omega XY=\Tr_\omega \omega(Y) X.
    \end{equation}
\end{defi}

These twisted traces have an intimate connection with short star products.  These were present in the physics literature in the work of Beem, Peelaers, and Rastelli \cite[\S 3.1]{beemDeformationQuantization2017}, but were introduced into the mathematical literature under this name by Etingof and Stryker \cite{etingofShortStarProducts2020}.  We call $\Tr_{\omega}$ nondegenerate with respect to a filtration $\mathscr{A}_0\subset \mathscr{A}_1\subset \cdots \subset \mathscr{A}$ if the bilinear form $(a,b)_{\Tr}=\Tr(ab)$ is nondegenerate in $\mathscr{A}_n$ for all $n$.  As shown by Etingof and Stryker \cite[Prop. 3.3]{etingofShortStarProducts2020}, based on a proposal of Kontsevich, a quantization that carries a nondegenerate trace will also carry a short star product.

One special way in which nondegeneracy can often appear is if there is a complex antilinear ring homomorphism $\antilinear\colon \mathscr{A}\to \mathscr{A}$ such that $\antilinear^2=\omega$. We say that $\Tr$ is {\bf positive} for $\antilinear$ if $\Tr_{\omega}(\antilinear (a)a)>0$ for all non-zero $a\in \mathscr{A}$; in this case, $\langle a,b\rangle=\Tr_{\omega}(\antilinear (a)b)$ is a positive definite inner product and nondegeneracy is manifest.    

In this paper, we will study these primarily in the case of abelian theories, though some of our results will also be interesting in the nonabelian case.  In particular, we will give a mathematical construction of the sphere trace $\operatorname{Tr}_\mass^G\colon \AHiggs\to \C$ using an imaginary rotation of the Weyl integration formula, assuming that the pair $(G,N)$ is good or ugly (Definition~\ref{enough-matter}) in the usual sense of 3-d gauge theories (Definition \ref{enough-matter}). If $G$ is a torus, this just means that the representation is faithful, but for nonabelian groups, it is a stronger requirement.    Since $\AHiggs$ is constructed mathematically as a noncommutative Hamiltonian reduction, one can think of this as a reduction of the unique twisted trace on the Weyl algebra.  While this construction has appeared in \cite[\S 2.7]{gaiottoSphereCorrelation2020} and is developed further in \cite{gaiottoSphereQuantization2023}, it seems to be new in the mathematics literature.  We prove the following:
\begin{theorem}[\mbox{\cref{th:reduced-positive}}]
    When the pair $(G,N)$ is conical (that is, good or ugly), the mass parameter $\mass$ and F-I parameter $\FI$ are both trivial and the moment map $T^*\C^n\to \fg_{\C}^*$ is flat,  the trace $\operatorname{Tr}_\mass^G$ is nondegenerate and positive.
\end{theorem}

It is a very interesting question for which values of $\mass,\FI$ this trace remains positive or just nondegenerate. One case where this is better understood is that of the quantized $A_n$ singularities, studied by Etingof, Klyuev, Rains, and Stryker \cite{etingofTwistedTraces2021}. These are examples of quantum Higgs and Coulomb branches for abelian gauge theories and thus should shed some light on the general case.  

Although the twisted trace on the Weyl algebra is unique up to scalar, this is not true of the algebra $\AHiggs$.  In particular, for modules in category $\cO$, we can take an appropriately twisted trace of the action on a module.    For the algebra $\AHiggs$, there is a natural class of modules to consider: the Verma modules.  The twisted traces of Verma modules are studied in \cite[\S 4.2]{etingofShortStarProducts2020}, but we can make more precise statements in the case of abelian gauge theories.   For generic FI parameters, these are simple, and the only simple modules in category $\mathcal{O}$ for these algebras (see Lemma~\ref{lem:generic-verma}).  The other of our main theorems is an expansion of the sphere trace  $\operatorname{Tr}_\mass^G$  in terms of the twisted traces on Verma modules, building on the predictions of \cite{gaiottoSphereCorrelation2020}.  This is explained in terms of hemispherical partition functions which limit to the character of Verma modules in \cite{bullimoreBoundariesVermas2021}.

In physical terms, these Verma modules correspond to supersymmetric vacua of the 3d quantum field theory. In more familiar mathematical terms, this means the fixed points of a torus action on the resolved Higgs branch, which is a hypertoric variety.  Intriguingly, our expansion has the appearance of an Atiyah-Bott localization formula.  By a theorem of Kamnitzer-McBreen-Proudfoot \cite[Th. 1.2]{kamnitzerQuantumHikita2021}, the space of all twisted traces on $\AHiggs$ is identified with the dual of quantum cohomology on the resolved Coulomb branch of this theory.  This suggests that the traces of Verma modules should be identified with fixed points and the sphere trace with equivariant integration over this Coulomb branch.  Another interesting observation is that the coefficients of the expansion of the sphere trace in terms of twisted traces of Verma modules for $\AHiggs$ are (up to sign) the twisted traces of 1 on Verma modules for the algebra $\ACoulomb$, as predicted in \cite[(1.5)]{gaiottoSphereCorrelation2020}.

While we leave mathematical consideration of the trace on $\ACoulomb$ to another time, it is notable that in the abelian case, our description of the sphere trace can be reinterpreted under 3-d mirror symmetry: the slice-projection theorem on Fourier transforms switches the Higgs-type integral we consider here to a Coulomb-type integral as in \cite[(2.17)]{gaiottoSphereCorrelation2020}.  

The study of twisted traces, especially those arising from partition functions of theories, is an intriguing new direction in the study of symplectic singularities.  The results here raise some interesting questions, such as:
\begin{itemize}
    \item For which automorphisms do quantizations of symplectic singularities always carry nondegenerate twisted traces? 
    \item Does every symplectic singularity carry at least one nondegenerate twisted trace? Etingof and Stryker consider this question for the automorphism they denote $s$, which would be $\onef$ in our notation, in \cite[\S 4]{etingofShortStarProducts2020} and show that a suitably generic twisted trace is nondegenerate when the symplectic singularity is a quotient of a vector space by a symplectic representation of a finite group or the nilcone of a simple Lie algebra.   
    \item Twisted traces can be thought of as a manifestation of a non-holomorphic $SU(2)$ action that arises naturally from the hyperk\"ahler structure on the Higgs and Coulomb branches of a 3d $\mathcal{N}=4$ (see the discussion in \cite[\S A]{beemDeformationQuantization2017}).  Can twisted traces help with finding hyperk\"ahler metrics or help us distinguish hyperk\"ahler varieties from those which are only complex symplectic?
    \item As mentioned above, the expansion of the sphere trace in terms of Verma modules has the form of an Atiyah-Bott formula on the dual cone.  Can this be made precise and extended to other cases?
\end{itemize}

\section{Sphere traces on Higgs branches}
\subsection{Weyl algebras}
On the Higgs side, we begin with $n$ hypermultiplets and then gauge by the action of a group $G$.  Mathematically, this can be expressed as a symplectic $\C$-vector space $(V,\Omega)$ whose dimension is $2n$.  To match the discussion in the introduction, we should take $V=T^*N$ with its induced $G$-action.
The local operators in the B-twist of this theory of $n$ hypermultiplets before gauging will be the Weyl algebra $\Weyl=\Weyl(V,\Omega)$, the $\mathbb{C}$-algebra freely generated by the elements of $V$, modulo the relation:
\begin{equation}
    uv-vu=\Omega(u,v).  
\end{equation}
We can choose a Darboux basis $x_1,\dots, x_n,y_1,\dots,y_n$ with the canonical commutation relations $[x_i,y_j]=\delta_{ij}$ and $[x_i,x_j]=[y_i,y_j]=0$.  Note that this implicitly writes $V=T^*N$ for a subspace $N$.  Note that pairing with the symplectic form induces an isomorphism $\kappa\colon V \to V^*$ satisfying $\Omega(u,v)=\langle u,\kappa(v)\rangle$ for the canonical pairing $\langle-,-\rangle\colon V\otimes V^*\to \C$.  This sends the Darboux basis above to itself up to sign: $\kappa(x_i)= -y_i^*,\kappa(y_i)=x_i^*$.

 Consider a connected compact group $G$ with a symplectic action on $(V,\Omega)$.  Note that we can also choose an inner product preserved by $G$ compatible with $\Omega$, making $V$ into a quaternionic vector space (that is, a flat hyperk\"ahler manifold), and $G$ is a subgroup of $USp(V,\Omega)=Sp(V,\Omega)\cap U(V)$, the subgroup that preserves the hyperk\"ahler structure.   Let $\No=N_{USp(V,\Omega)}^{\circ}(G)$ be the connected component of the identity in the normalizer of $G$.  The quotient $\No/G=\tilde{F}$ is the flavor symmetries of this theory.  We will also find it useful to consider the group $\To\subset \No$ generated by $G$ and a fixed maximal torus of $\No$; this is the same as the preimage of a maximal torus $F\subset \tilde{F}$.  \Ben{This looks  a little weird, but it's so we match the usual use of $F$ in the abelian section.}

The group $\No$ acts on $\Weyl(V,\Omega)$, and thus the Lie algebra $\no$ acts on it as well.  This action is inner: for $\LE\in \no$, let $\hat{\mu}(\LE)=\frac{1}{2}m(1\otimes \kappa^{-1}(\LE_V))$ where $\LE_V\in V\otimes V^*$ is the action of $\LE$ on $V$ and $m\colon V\otimes V \to\Weyl(V,\Omega) $ given by multiplication.  In particular, if we diagonalize the action of $\LE$ on $V$, then we can choose the Darboux basis so that $\LE x_i=-a_ix_i, \LE y_i=a_iy_i$, so \begin{align}
    \LE_V&=\sum_{i=1}^n -a_ix_i\otimes x_i^*+a_iy_i\otimes y_i^* & \kappa(\LE_V)&=\sum_{i=1}^n a_ix_i\otimes y_i+a_iy_i\otimes x_i\\ \hat{\mu}(\LE )&=\sum_{i=1}^n \frac{1}{2}a_i(x_iy_i+y_ix_i)=\sum_{i=1}^n a_i\bigg(y_ix_i+\frac{1}{2}\bigg).
\end{align}
If $G$ is abelian, then we can make one choice of Darboux coordinates in which all $\beta$ are diagonal, whereas if $G$ is non-abelian, this is clearly not possible unless the action factors through an abelian quotient.  
\begin{lemma} For any $w\in \Weyl(V,\Omega)$, we have 
\begin{equation}
[\hat{\mu}(\LE ), w] =\LE  \cdot w.
\end{equation}
\end{lemma}
\begin{proof}
We can choose a Darboux basis as above and verify:
\begin{align}
    [\hat{\mu}(\LE ),x_j]&=\sum_{i=1}^n \frac{1}{2}a_i[x_iy_i+y_ix_i,x_j]\\
    &=\frac{1}{2}a_j[x_jy_j+y_jx_j,x_j]\\
    &=\frac{1}{2}a_jy_jx_j^2-x_j^2y_j=-a_jx_j
\end{align}
and similarly for $y_i$ with signs reversed.
\end{proof}
It's also useful to write this equation as $\hat{\mu}(\LE )x_j=x_j(\hat{\mu}(\LE )-a_i)$
Note that this implies that for any $\LE \in \no_{\mathbb{C}}$ and $h=\exp(\LE )$, we have $h\cdot w=\exp(\hat{\mu}(\LE ))w\exp(-\hat{\mu}(\LE ))$.

Throughout, we will want to fix two quantities:
\begin{enumerate}
    \item 
Fix $\mass \in i\no$ which we assume acts invertibly on $V$, and let $\hat{\mu}=\hat{\mu}(\mass )\in \Weyl$; in physical terms, this is the mass parameter of the theory.  
\item We also fix a character $\FI \colon \no\to \mathbb R$; in physical terms, this is the FI parameter of the corresponding 3d theory. 
\end{enumerate}  We'll usually want to consider the shifted moment map $\hat{\mu}_{\FI}(X)=\hat{\mu}(X)+\FI(X)$, with $\hat{\mu}_{\FI}=\hat{\mu}_{\FI}(\mass)$.
\Ben{Do I want plus or minus?  Sigh.  I guess Davide's paper uses plus. }
Since $\No$ is compact, the action of $\mass$ on $V$ has only non-zero real eigenvalues.  Furthermore, if we let $V^{\pm}$ be the sum of the positive and negative eigenspaces for $2\pi \mass$, respectively, then these subspaces are Lagrangian.  Consider the module $\mathbb{V}=\Weyl/\Weyl V^{+}$; in physics terms, we consider a vacuum vector $|0\rangle$ and let $V^{+}$ be annihilation operators, while $V^-$ are viewed as creation operators.  Let $g=\exp(2\pi \mass )\in \No_{\mathbb{C}}$.

Let $\onef\colon \Weyl \to \Weyl$ be the endomorphism induced by scalar multiplication by $-1$ on $V$; we use the same symbol to denote the induced endomorphism of $\mathbb{V}$.  
\begin{theorem}\label{thm:Weyl-trace}
The functional $\operatorname{Tr}_\mass \colon \Weyl \to \C$ defined by \begin{equation}
    \operatorname{Tr}_\mass (w)=\operatorname{Tr}(\mathbb{V};\onef  \exp(2\pi \hat{\mu}_{\FI}) w  )
\end{equation}
is a well-defined twisted trace on $\Weyl$ for the automorphism given by the action of $\onef g^{-1}$ and up to scalar multiplication, this is the unique such twisted trace. 
\end{theorem}
As with several results that we present here, this fact has been noted several times in the physics literature, for example, in \cite[\S 3]{gaiottoSphereCorrelation2020}.  Also note that the connection of this twisted trace to star-products on the Weyl algebra is discussed in \cite[Ex. 2.7]{etingofShortStarProducts2020}.
Since $g$ uniquely determines $\mass=\frac{1}{2\pi}\log g$, we can also write $\Tr_g=\Tr_{\mass}$.  
\begin{proof}
{\it $\operatorname{Tr}_\mass$ is well-defined}: Choose our Darboux basis so that $x_1,\dots, x_n$ span $V^-$ and $y_1,\dots, y_n$ span $V^+$; let $a_i\in \mathbb R$ be the eigenvalue $[\hat{\mu}_{\FI},y_i]=a_iy_i$.  The representation $\mathbb{V}$ is isomorphic to $\C[x_1,\dots, x_n]$ by action on the vacuum vector.  Thus, $2\pi \hat{\mu}$ acts on $x_1^{u_1}\cdots x_n^{u_n}|0\rangle$ by
\begin{align}
    2\pi \hat{\mu}\cdot x_1^{u_1}\cdots x_n^{u_n}|0\rangle&=2\pi x_1^{u_1}\cdots  x_n^{u_n}( \hat{\mu}-\sum_{j=1}^na_ju_j)|0\rangle\\
    &=-2\pi x_1^{u_1}\cdots  x_n^{u_n}\sum_{j=1}^na_j(u_j+\frac{1}{2})|0\rangle
\end{align}
Thus, $\exp(2\pi \hat{\mu})$ is a diagonalizable endomorphism which acts on  $x_1^{u_1}\cdots x_n^{u_n}|0\rangle$ by $\exp(2\pi \sum_{j=1}^na_j(u_j+\frac{1}{2}))$.  In particular, the trace of $w$ only depends on its diagonal entries in this basis and thus only depends on the projection of $w$ to the subalgebra $\C[x_1y_1,\dots, x_ny_n]$.  For a polynomial $f(x_1y_1,\dots, x_ny_n),$ we have that:
\begin{multline}
 \onef  \exp(2\pi  \hat{\mu}_{\FI})f(x_1y_1,\dots, x_ny_n) x_1^{u_1}\cdots x_n^{u_n}|0\rangle\\=(-1)^{u_1+\cdots +u_n}f(-u_1,\dots, -u_n)\exp(2\pi \big(\FI(\mass)-\sum_{j=1}^na_j(u_j+\frac{1}{2})\big)) x_1^{u_1}\cdots x_n^{u_n}|0\rangle
\end{multline}
Thus, we have that
\begin{multline}\label{eq:Tr-LE-sum}
      \Tr_{\mass}(f(x_1y_1,\dots, x_ny_n) )= \operatorname{Tr}(\mathbb{V};  \onef \exp(2\pi  \hat{\mu}_{\FI})f(x_1y_1,\dots, x_ny_n) )\\=\sum_{\mathbf{u}\in \Z_{\geq 0}^n}(-1)^{u_1+\cdots +u_n}f(-u_1,\dots, -u_n)\exp(2\pi \big(\FI(\mass)-\sum_{j=1}^na_j(u_j+\frac{1}{2})\big))
\end{multline}
which converges by the ratio test since $2\pi a_j>0$ for all $j$.  Thus, we have a well-defined functional on $\Weyl$.

{\it $\operatorname{Tr}_\mass$ is a twisted trace}: \begin{align}
    \operatorname{Tr}_\mass (w_1w_2)&=\operatorname{Tr}(\mathbb{V}; \onef \exp(2\pi \hat{\mu}_{\FI})w_1w_2)\\ &=(-1)^{w_2}\operatorname{Tr}(\mathbb{V}; \exp(w_2\onef(2\pi \hat{\mu}_{\FI})w_1))\\
    &=\operatorname{Tr}_\mass ((\onef g\cdot w_2)w_1)
\end{align}

{\it $\operatorname{Tr}_\mass$ is unique up to scalar}:  Let $\operatorname{Tr}'$ be any other twisted trace for the same automorphism.  The difference $\operatorname{Tr}'-\frac{\operatorname{Tr}'(1)}{\operatorname{Tr}_{\mass}(1)}\operatorname{Tr}_{\mass}$ is still a twisted trace and vanishes on the identity.  Thus, it suffices to show that if $\operatorname{Tr}'(1)=0$, then $\operatorname{Tr}'(w)=0$ for all $w\in \Weyl$.  We have $g^{-1}\cdot x_iy_i=x_iy_i$, so $\operatorname{Tr}'([x_iy_i,w])=0$ for all $w\in \Weyl$.  This shows that $\operatorname{Tr}'(w)$ is unchanged by projection to $\mathbb{C}[x_1y_1,\dots, x_ny_n]$.  Let $r=x_1^{u_1}y_1^{u_1}\cdots x_n^{u_n}y_n^{u_n}$ be a general monomial in this subring and choose an index $j$ such that $u_j>0$.  We can write $r=x_j^{u_j}y_j^{u_j}r'$ where $r'$ commutes with $x_j,y_j$.  Now, assume that $u_j$'s are chosen to have minimal sum such that $\operatorname{Tr}'(r)\neq 0.$  Then we have
\begin{align}
    \operatorname{Tr}'(r)&=\frac{1}{1-e^{2\pi a_j} }( \operatorname{Tr}'(x_j^{u_j}y_j^{u_j}r'-e^{2\pi a_j} x_j^{u_j-1}y_j^{u_j}r'x_j)\\ &=\frac{u_j}{1-e^{2\pi a_j} }\operatorname{Tr}'(x_j^{u_j-1}y_j^{u_j-1}r')\\
    &=0
\end{align}
by the assumption of minimality.  We had assumed that $\operatorname{Tr}'(r)\neq 0$, so we have arrived at a contradiction.  
\end{proof}
Theorem \ref{thm:Weyl-trace} shows that the dependence of $\Tr_{\mass}$ on $\mass$ is completely captured by the value on 1.  
Later, we'll use the geometric series expansions, valid when $x<0$ and $x>0$ respectively:
\begin{align}\label{eq:sec-sum1}
	\frac{1}{2}\sech(x)&=e^{x}-e^{3x}+e^{5x}-\cdots= e^{-x}-e^{-3x}+e^{-5x}-\cdots\\
	\label{eq:sec-sum2}\frac{1}{2}\sech^{(r)}(x)&=e^{x}-3^re^{3x}+5^re^{5x}-\cdots
=(-1)^re^{-x}-(-3)^re^{-3x}+(-5)^re^{-5x}-\cdots\\
\label{eq:sec-sum3}
	\frac{1}{2}\csch(x)&=-e^{x}-e^{3x}-e^{5x}+\cdots= e^{-x}+e^{-3x}+e^{-5x}+\cdots\\
	\label{eq:sec-sum4}	\frac{1}{2}\csch^{(r)}(x)&=-e^{x}-3^re^{3x}-5^re^{5x}-\cdots
=(-1)^{r}e^{-x}+(-3)^re^{-3x}+(-5)^re^{-5x}+\cdots
\end{align}
From \eqref{eq:Tr-LE-sum} and \eqref{eq:sec-sum1}, we see that
\begin{equation}\label{eq:Tr-LE-1}
    \operatorname{Tr}_\mass(1)=e^{2\pi \FI(\mass)}\prod_{j=1}^n(e^{-\pi a_j}-e^{-3\pi a_j}+e^{-5\pi a_j}+\cdots )=e^{2\pi \FI(\mass)}\prod_{j=1}^n\frac{1}{2}\sech(\pi a_j)
\end{equation}
In terms of the scalars $A_j=\exp(\pi  a_j)$, we have 
\begin{equation}\label{eq:Tr-LE-3}
    \operatorname{Tr}_\mass(1)=e^{2\pi \FI(\mass)}\prod_{j=1}^n(A_j^{-1}-A_j^{-3}+A_j^{-5}+\cdots )=e^{2\pi \FI(\mass)}\prod_{j=1}^n\frac{1}{A_j+A_j^{-1}}
\end{equation}
Since it will be useful later, let us note that by \eqref{eq:sec-sum2}:
\begin{multline}
	\label{eq:Tr-LE-2}
    \operatorname{Tr}_\mass\left(\prod_{j}(x_jy_j-\nicefrac{1}{2})^{u_j}\right)=e^{2\pi \FI(\mass)}\prod_{j=1}^n((\nicefrac{-1}{2})^{u_j}e^{-\pi a_j}-(\nicefrac{-3}{2})^{u_j}e^{-3\pi i a_j}+(\nicefrac{-5}{2})^{u_j}e^{-5\pi a_j}+\cdots )
    \\=e^{2\pi \FI(\mass)}\prod_{j=1}^n\frac{1}{2^{u_j+1}}\sech^{(u_j)}(\pi a_j)
\end{multline}

While the interpretation as a trace does not make sense at $\mass=0$, the right-hand side of the equations (\ref{eq:Tr-LE-1}--\ref{eq:Tr-LE-2}) still makes sense and defines a twisted trace; we'll discuss an alternate construction that gives the trace in this case below.

In general, we can think of the computation of the twisted trace of any monomial in $\Weyl$ as a twisted version of Wick contraction:  for any $v\in V$ and any $x\in \Weyl$, the operator $\gamma=1+g$ is invertible, and we have that  \begin{equation}
	\Tr_{\mass}(xv)=\Tr_{\mass}(x\gamma^{-1}(v)+g\gamma^{-1}v)=\Tr_{\mass}(x\gamma^{-1}v-\gamma^{-1}vx), 
\end{equation} where the commutator can be computed using Wick's theorem.  This reduces to computing the trace of an element of lower degree.  

We can think of $\operatorname{Tr}_\mass (w)$ as a function of $\mass$; this is only well-defined on the subset of $\no$ where elements have the same positive eigenspaces as $\mass$, but we can extend by analytic continuation as the formula \eqref{eq:Tr-LE-1} above shows; in order to make sense of this formula, we should think of $a_j$ as a linear function of the element $\mass$. Note that this also shows that this analytic continuation agrees with the definition of $\operatorname{Tr}_\mass$ using a different vacuum module.  

Thus, we can consider the derivative of this trace for a fixed $w$.  For any element $\LE\in \no_{\mathbb{C}}$, thought of as a tangent vector, the Lie derivative is given by 
\begin{align}
    \mathcal{L}_{\LE}\operatorname{Tr}_\mass(w)&=
    \operatorname{Tr}(\mathbb{V}; \onef \mathcal{L}_{\LE}\exp(2\pi \hat{\mu}_{\FI}) w) \notag\\
    &=\operatorname{Tr}(\mathbb{V}; \onef \exp(2\pi \hat{\mu}_{\FI}) \hat{\mu}_{\FI}(\LE) w)\label{eq:directional-derivative}\\
    &=\operatorname{Tr}_\mass(\hat{\mu}_{\FI}(\LE)w)\notag
\end{align}

\subsection{Hilbert space realization} There is a second realization of this twisted trace which is more physically natural; this is a slight generalization of the construction given in \cite[\S 3.1]{gaiottoSphereQuantization2023}.  Let $\alpha\colon \Weyl \to \operatorname{Op}(\mathcal{S}(\C^n))$ be the ring homomorphism to the continuous operators on the Schwartz space $\mathcal{S}(\C^n)$, where we let $x_j$ act by the multiplication operators $z_j$, and $y_j\mapsto -\frac{\partial}{\partial z_j}$.

We also have the commuting complex conjugate action 
\begin{equation}
	\bar{\alpha}(x_j)=\bar{z}_j \qquad \bar{\alpha}(y_j)=-\frac{\partial}{\partial \bar z_j}.
\end{equation} Note that under adjunction, $\alpha(w)^{\dagger}=\bar{\alpha}(w^{\dagger})$ where $\dagger\colon \Weyl\to \Weyl$ is the $\C$-antilinear anti-automorphism sending $x_j\mapsto x_j, y_j\mapsto -y_j$.  Furthermore $\hat{\mu}_{\FI}^{\dagger}=-\hat{\mu}_{\FI}$ since $\FI(\mass)\in i\mathbb R$. 
Let $g=\exp(2\pi \alpha(\hat{\mu}_{\FI}))$; 
the adjoint operator is given by $g^{\dagger}=\exp(-2\pi\bar{\alpha}(\hat{\mu}_{\FI}))$.

Since any element of $L^2(\C^n)$ be be approximated arbitrarily well by polynomials in $z_i,\bar{z}_i$, it suffices to calculate how any operator acts on these polynomials.  In particular, we have 
\begin{equation}\label{eq:muz-action-1}
    \alpha(\hat{\mu}_{\FI})\cdot z_1^{u_1}\bar{z}_1^{v_1}\cdots z_n^{u_n}\bar{z}_n^{v_n}=(\FI(\mass)-\sum_{j=1}^n a_j(u_j+\frac{1}{2}))z_1^{u_1}\bar{z}_1^{v_1}\cdots z_n^{u_n}\bar{z}_n^{v_n}
\end{equation}
\begin{equation}\label{eq:g-action-1}
    g\cdot z_1^{u_1}\bar{z}_1^{v_1}\cdots z_n^{u_n}\bar{z}_n^{v_n}=e^{2\pi \FI(\mass)}\left(\prod_{j=1}^n A_j^{-2u_j-1}\right)z_1^{u_1}\bar{z}_1^{v_1}\cdots z_n^{u_n}\bar{z}_n^{v_n}
\end{equation}
\begin{align}
g\alpha(y_j)g^{-1}&=A_j^2\alpha(y_j)& g\alpha(x_j)g^{-1}&=A_j^{-2}\alpha(x_j)\\
g\bar{\alpha}(y_j)g^{-1}&=\bar{\alpha}(y_j)& g\bar\alpha(x_j)g^{-1}&=\bar\alpha(x_j)
\end{align}
for scalars $A_j=\exp(\pi  a_j)$.  Let $g^{p}=\exp(2p\pi \alpha(\hat{\mu}_{\FI}))$.

 The Gaussian function $|1\rangle=e^{-|z|^2}$ plays an important role in our construction.  By \cref{eq:g-action-1}, we have that:
\begin{equation}
    g |1\rangle=e^{2\pi \FI(\mass)}\prod_{j=1}^n\sum_{k=0}^{\infty}(-1)^k \frac{A_{j}^{-2k-1}z_j^k\bar{z}_j^k}{k!}=e^{2\pi\FI(\mass)}\prod_{j=1}^n\frac{1}{A_j}e^{-z_j\bar{z}_j/A_j^2} 
\end{equation} 
 Consider the inner product on $L^2(\C^n)$ given by 
 \[\langle f|f'\rangle =\frac{2^n}{\pi^n}\int_{\C^n}
x \bar{f}(z,\bar{z}) f'(z,\bar{z})dzd\bar{z}\]
We let $\langle 1 |$ be the functional of inner product with $|1\rangle$.  Note that we have normalized this inner product so that $\langle 1|1\rangle=1$.  An important calculation is that:
\begin{equation}\label{eq:gp-integral}
     \langle 1|g^p |1\rangle=\frac{2^n}{\pi^n}e^{2\pi p\FI(\mass)}\int_{\C^n}\prod_{j=1}^n\frac{1}{A_j^p}e^{-(1+1/A_j^{2p})z_j\bar{z_j}}dzd\bar{z}=e^{2\pi p\FI(\mass)}\prod_{j=1}^n\frac{1}{A_j^p+A_j^{-p}}.
\end{equation}
The states $|1\rangle,\langle 1|$ have the remarkable property that:
 \begin{align}
 \label{eq:gaussian-r}	\alpha(x_j)|1\rangle&=\bar \alpha(y_j)|1\rangle & \alpha(y_j)|1\rangle&=\bar \alpha(x_j)|1\rangle\\
 	\label{eq:gaussian-l} \langle 1|\alpha(x_j)& =- \langle 1|\bar \alpha(y_j)&\langle 1|\alpha(y_j)& =- \langle 1|\bar \alpha(x_j)
 \end{align}
 Combining these properties, we find that:
\begin{lemma}\label{lem:trace-pairing}
	We have an equality $\Tr_{\mass}= \langle 1|g\alpha(w)| 1\rangle=\langle 1|\alpha(w)g| 1\rangle$.
\end{lemma}
This construction may seem slightly unmotivated; of course, there is physical motivation behind it, but it also has a simple explanation in terms of the usual geometric interpretation of trace where we place operators on a circle.  Imagine a loop where we put $\langle 1|g$ and $|1\rangle$ at the far left and right ends, respectively; the images of operators under $\alpha$ on the bottom; and images under $\bar{\alpha}$ at the top, as shown below:
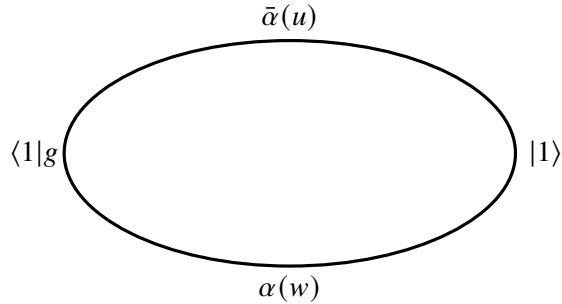
\[\begin{tikzpicture}
\draw[very thick] (0,0) ellipse (3cm and 1.5cm);
\node at (0,1.8) {$\bar{\alpha}(u)$}; \node at (-3.4,0) {$\langle 1|g$}; \node at (0,-1.8) {${\alpha}(w)$}; \node at (3.4,0) {$|1\rangle$}; \end{tikzpicture}\]
We can then evaluate this picture to a scalar by multiplying the operators on the top and bottom segments in order from left to right and sandwiching this product between $ \langle 1|g$ and $|1\rangle$.  Thus, the formula for the trace of \cref{lem:trace-pairing} is obtained by taking $u=1$ in the picture above.  
We can reorder pairs of operators on the top and bottom of the circle without changing this evaluation since the images of any operators under $\alpha$ and $\bar{\alpha}$ commute.  The equations \crefrange{eq:gaussian-r}{eq:gaussian-l} show how we move operators around the left and right ends of the circle.  We can thus view the proof below as factoring an operator at the bottom of this circle and moving one of the factors around the circle to the other side of the bottom segment.  We obtain a twisted trace, since this operation induces an automorphism of the underlying algebra $\Weyl$.  
\begin{proof}
Comparing \cref{eq:Tr-LE-3} and \cref{eq:gp-integral}, we find that $\Tr_{\mass}(1)=\langle 1|g| 1\rangle$.  Thus, we need only to show that $\langle 1|g\alpha(w)| 1\rangle$ is a twisted trace for the automorphism $\omega=\onef g^{-1}$.  
	Consider $w=ux_j$.  Note that we have $\omega(x_j)=A_jx_j$.  
	\begin{multline}
		\langle 1|g\alpha(ux_j)| 1\rangle = \langle 1|g\alpha(u)\bar{\alpha}(y_j)| 1\rangle= \langle 1|\bar{\alpha}(y_j)g\alpha(u) | 1\rangle\\=-\langle 1|{\alpha}(x_j)g\alpha(u) | 1\rangle=\langle 1|g{\alpha}(\omega(x_j)u) | 1\rangle
	\end{multline}
	Similarly, for $y_j$, we have $\omega(y_j)=-A_jy_j$, and 
	\begin{equation}
		\langle 1|g\alpha(uy_j)| 1\rangle=\langle 1|g\alpha(u)\bar{\alpha}(x_j)| 1\rangle=\langle 1|g\alpha(\omega(y_j)u) | 1\rangle.
	\end{equation}
	This suffices to prove the twisted trace property for any factorization $w=ab$.  The equality $\langle 1|g\alpha(w)| 1\rangle=\langle 1|\alpha(w)g| 1\rangle$ holds since unless $w$ commutes with $g$, any twisted trace for $\omega$ will vanish on $w$.    
\end{proof}
This gives another realization of the trace $\Tr'$ which makes sense for $\mass=0$ and does not depend on a choice of Lagrangian subspace.  Furthermore, the formula we have given makes sense for any $\mass\in \mathfrak{h}_{\C}$, but the imaginary part will not change the trace, since $|1\rangle$ is invariant under the compact group $G$.  The case where $\mass=0$ and $\FI=0$ (so $g=1$) has a special positivity property:
Observe that the map $w\mapsto \alpha(w)| 1\rangle$ is a linear embedding of $\Weyl$ in the Hilbert space $L^2(\C^n)$ since the images \[\alpha(x_1^{a_1}\cdots x_n^{a_n}y_1^{b_1}\cdots y_n^{b_n})|1\rangle=z_1^{a_1}\cdots z_n^{a_n}\bar{z}_1^{b_1}\cdots \bar{z}_n^{b_n}|1\rangle\] are all linearly independent.  Thus, it induces a Hermitian inner product on $\Weyl$, given by 
\begin{equation}
	\langle a,b\rangle=\prod_{i=1}^n |\langle 1|\bar\alpha(a^{\dagger})\alpha(b)| 1\rangle=\langle 1|\alpha(\antilinear(a)b)| 1\rangle=\Tr_{\mass}(\antilinear(a)b)
\end{equation}
where $\antilinear$ is the anti-linear endomorphism defined by 
\begin{equation}\label{eq:sigma}
	\antilinear(x_j)=-y_j\qquad \antilinear(y_j)=x_j
\end{equation}
This is precisely the statement that:
\begin{lemma}\label{lem:Weyl-positive}
	If  $\mass=0,\FI=0$, the trace $\Tr_{\mass}$ is positive for $\antilinear$. 
\end{lemma}

\subsection{Hamiltonian reduction}
The discussion above has only covered the Weyl algebra; we now wish to consider how this algebra of operators is affected by gauging by a group $G$.  As before, we match the discussion in the introduction, we take $V=T^*N$. However, we do not want to assume that such an invariant Lagrangian subspace $N$ exists, since the results of this section hold without this assumption (which as noted before, will always hold in the abelian case).

Consider a character $\FI\colon \mathfrak{g}\to \C$.  We will consider the Hamiltonian reduction (that is, gauging) of $\Weyl$ by $G$:
\[\AHiggs=\big(\Weyl/\sum_{X\in \mathfrak{g}} \Weyl( \hat{\mu}_{\FI}(X))\big)^G.\] There is a functor of Hamiltonian reduction 
\[\red\colon \Weyl\mmod\to \AHiggs\mmod \qquad \red(M)=M/\sum_{X\in \mathfrak{g}} ( \hat{\mu}_{\FI}(X))M.\]
Note that the twisted trace $\operatorname{Tr}_\mass$ does not descend to this reduction, but as discussed in the introduction, we can construct an ``averaged'' version of it.  

We assume that the element $\mass\in \no$ chosen before commutes with $\mathfrak{g}$. We can then choose Cartan subalgebras  $\ft_{\no}\supset \ft$ of $\no$ and $\mathfrak{g}$ such that $\mass\in \ft_{\no,\mathbb{C}}$; let $z_1,\dots, z_r$ be orthonormal coordinates on $\ft$, which we can extend to orthonormal coordinates on $\ft_{\no}$.  
Consider the coset $\ft_{\mathbb{C}}+\mass$.   This is a complex subspace and contains the two real subspaces $\ft+\mass,i\ft+\mass$.  For a given $w\in \Weyl$, we'll consider the closed holomorphic $r$-form $\operatorname{Tr}_\mass(w)dz_{1}\wedge \cdots\wedge dz_{r}$.  
Let $\roots$ denote the set of positive roots of $G$, considered as functions on $\ft_{\no,\C}$.

\begin{defi}\label{enough-matter}
We say that the pair $(G,V)$ is {\bf conical} if for all $\LE\in \ft,$ we have $\sum_{j=1}^n |a_j(\LE)| > 2\sum_{\alpha\in \roots} |\alpha(\LE)|$.  If we have fixed a Lagrangian subrepresentation $N$, then we say that $(G,N)$ is conical in this circumstance. 
\end{defi}

In the usual classification of 3-d theories, being conical corresponds to being good or ugly (for example, the Weyl algebra $\Weyl$ corresponds to an ugly theory where the Coulomb branch is $\C^{2n}$).
This definition is most often made by only requiring that this condition only holds for integral $\LE$, but this is equivalent.
\begin{lemma}
    The function \begin{equation}\label{eq:Delta}
	\Delta(\LE)=\frac{1}{2}\sum |a_j(\LE)| -\sum_{\alpha\in \roots} |\alpha(\LE)|
\end{equation}	
is positive on all non-zero $\LE\in \ft$ if and only if it is positive on all non-zero integral $\LE\in \ft_{\Z}$.
\end{lemma}
\begin{proof}
    If $\Delta$ attains a negative value, we can approximate it arbitrarily well by $\frac{1}{m}\LE$ for some integer $m\in\Z$ and $\LE$ which is integral, and by the continuity of $\Delta$, once this this approximation is good enough, we will have $\Delta(\LE)<0$.  

    Thus, we may assume that $\Delta(\LE)$ is non-negative for all $\LE$. Since $\Delta(\LE)$ is piecewise linear on regions cut out by a finite number of hyperplanes, its zeros are a finite union of subsets defined by equalities and inequalities for these hyperplanes.  Furthermore, this region is closed under multiplication by positive real numbers.  Thus, unless it is just the origin, it contains a ray cut out by rational hyperplanes, which thus contains an integral point.  
\end{proof}

\begin{remark}
It's worth noting that we can interpret the quantity $\Delta(\LE)$ as the conformal dimension of the corresponding monopole operator, that is, half\footnote{The factor of 2 appears since physicists prefer to index representations of $SU(2)$ by $\frac{1}{2}\Z$, and mathematicians prefer $\Z$.} its degree in the natural grading on $\ACoulomb$ specialized at $\hbar=0$.  Thus being conical is equivalent to the corresponding $\C^{\times}$ action on the Coulomb branch contracting it to a point.  It is not clear how this connects to the content of this paper, but it is an intriguing observation.	
\end{remark}

Let $\Van(\LE)=\prod_{\alpha\in \roots} 2\sin (\pi \alpha( \LE))$ where, as before, $\roots$ is the set of positive roots of $G$.
\begin{defi}\label{def:Higgs-trace}  The {\bf sphere trace} of the quantum Higgs branch $\AHiggs$ is defined by the integral 
\begin{equation}\label{eq:Higgs-trace}
    \operatorname{Tr}_\mass^G(w)=\int_{\LE\in i \mathfrak{{t}}+\mass} \Van(\LE) \operatorname{Tr}_{\LE}(w)dz_{1}\wedge \cdots\wedge dz_{r}
\end{equation}
if this integral converges.  
\end{defi}
Note that we have omitted $\FI$ from the notation of the trace, since when we vary $\FI$, we change the underlying algebra, while varying $\mass$ gives different twisted traces on the same algebra.

As discussed in the Introduction, this trace first appeared as a correlation function of a 3d $\mathcal{N}=4$ SUSY theory; see \cite[\S 2.7]{gaiottoSphereCorrelation2020} and \cite[\S 4]{gaiottoSphereQuantization2023}.  
It is more natural to think of this as an integral over the Cartan complement $P=\exp(i\mto)\subset \To_{\C}$.  This follows from the same proof as the usual Weyl integration formula.  

\begin{lemma} \label{lem:Weyl-integration}Let $P_{\mass}=\exp(i\fg+\mass)$; we can write
\[\operatorname{Tr}_\mass^G(w)= \int_{P_\mass} \operatorname{Tr}_{g}(w)dg\]
where $dg$ denotes the volume form on this submanifold induced by the Cartan-Killing metric.\end{lemma}
  The term $\Van(\LE)$ thus arises from the Weyl intergration formula---it reflects the volume of the different conjugacy classes for the action of $G$ on $P_{\mass}$. 

\begin{proof}
We have a natural map $p\colon G/T\times_W (\ft+\mass)\to P_{\mass}$ given by $(gT,t)\mapsto g\exp(2\pi t)g^{-1}$. 
For any function $f\colon P_{\mass}\to \mathbb R$ invariant under conjugation by $G$, we can write \[\int_{P_\mass} f(g)dg=\frac{1}{\#W}\int_{G/T \times i\ft} f(g)p^*dg\]
where $dg$ denotes the volume form on this submanifold induced by the Cartan-Killing metric. We can compare this to the volume form $dg'$ for the induced metric on $G/T\times i\ft$.  These differ by $p^*dg=\det(J)dg'$ where $J$ is the Jacobian determinant of the map $p$ written in an orthonormal basis. This is a product of terms coming from the different root $SU(2)$'s; the image of the matrices $X=\big[\begin{smallmatrix}
0 & -1\\
1& 0
\end{smallmatrix}\big]$ and $Y=\big[\begin{smallmatrix}
0 & i\\
i& 0
\end{smallmatrix}\big]$ under these $SU(2)$'s give a basis of $\fg/\ft$.  
Thus, we only need to do this calculation for $SU(2)$.  Consider the matrix $A=  \big[\begin{smallmatrix}
\exp(\pi i a) & 0\\
0& \exp(-\pi i a)
\end{smallmatrix}\big].$ Recall that $a\in i\mathbb R$, so this is a real symmetric matrix. We then have \begin{align*}
	\frac{d}{dt}(e^{tX}Ae^{-tX})\big \vert_{t=0}&=[X,A]=\begin{bmatrix}
0& 2i\sin(\pi a) \\
2i\sin(\pi a)& 0 
\end{bmatrix}\\
	\frac{d}{dt}(e^{tY}Ae^{-tY})\big \vert_{t=0}&=[Y,A]=\begin{bmatrix}
0& 2\sin(-\pi a) \\
2\sin(-\pi a)& 0 
\end{bmatrix}
\end{align*}
so the Jacobian is the product of factors 
\[(2\sin(\pi a))(2\sin(-\pi a))=4\sinh^2(i\pi a).\] For a general reductive group, this becomes $\Van(\LE)=\prod_{\alpha\in \roots}4\sinh^2(i\pi  \alpha( \LE))$. 
\end{proof}

\begin{theorem}\label{th:integral-twisted-trace}
If $(G,V)$ is conical, then this integral is convergent and only depends on the image of $w$ in $\AHiggs$.  Thus, it defines a twisted trace for the action of $g=\exp(2\pi \mass)$.
\end{theorem}
\begin{proof}
{\bf The integral is well-defined:} As usual, we only need to consider the case where $w=\prod_{j}(x_jy_j-\nicefrac{1}{2})^{u_j}$.  In this case, the integral becomes  \[\int_{i\mathfrak{{h}}+\mass} \Van(\LE)e^{2\pi \FI(\LE)}\prod_{j=1}^n\frac{1}{2^{u_j+1}}\sech^{(u_j)}(\pi a_j(z_1,\dots, z_r))dz_{1}\wedge \cdots \wedge dz_{r}.\]
Since $(G,V)$ is conical, the function $m(\eta)=\sum_{j=1}^n |a_j(\eta)| - 2\sum_{\alpha\in \roots} |\alpha(\eta)|$ is positive valued on the space of $\eta\in \ft$ of norm 1, and attains a minimum $\xi$.   
There is a real number $C>0$ such that 
\[ |\Van(\LE) \operatorname{Tr}_{\LE}(w)| < \frac{C}{e^{\xi |\beta|}}\]
so this integral is well-defined.

{\bf $ \operatorname{Tr}_\mass^G(w)$ only depends on the image of $w$ in $\AHiggs$:} To show this, we need only check that $\operatorname{Tr}_\mass^G(w\hat{\mu}_{\FI}(\LE'))=0$ for all $w\in \Weyl,\LE'\in \fg$.  We can think of $i\LE'$ as a vector field on $\To_{\C}$, and by Cartan's magic formula:
\[\operatorname{Tr}_{g}(w\hat{\mu}_{\FI}(\LE'))dg=-i\mathcal{L}_{i\LE'}\operatorname{Tr}_{g}(w)dg=-id\iota_{i\LE'}(\operatorname{Tr}_{g}(w)dg)\]
So by Stokes' theorem,  
$\operatorname{Tr}_\mass^G(w\hat{\mu}_{\FI}(\LE'))=0$.

{\bf $\operatorname{Tr}_\mass^G$ is a twisted trace:} This follows immediately from applying the twisted trace property inside the integral.  
\end{proof}
This defines a canonical twisted trace for the automorphism induced by $g$ on the Hamiltonian reduction, but unlike the original Weyl case, it is not unique.

We can extend our Hilbert space realization of this trace to the reduced case.  We can do this by considering again the function $|1\rangle=e^{-|z|^2}$.  This is obviously invariant under the compact group $G$, but not for the full group $G_{\C}$.  
Using \cref{lem:trace-pairing}, we can write
\[\operatorname{Tr}_\mass^G(w)= \int_{P_\mass}  \langle 1|g\alpha(w)| 1\rangle dg=\int_{i\ft+\mass}\langle 1|e^{2\pi \alpha(\hat{\mu}_{\FI}(\LE))}\alpha(w)| 1\rangle d\LE.\]
Since $| 1\rangle$ is invariant under $G$, there is no effect in also integrating over the compact group $G$, and it will be occasionally useful to write $\operatorname{Tr}_\mass^G(w)= \int_{G_{\C,\mass}}  \langle 1|g\alpha(w)| 1\rangle dg$.
Exchanging the order of integration, we can write this as the inner product of $\langle \mathbbm{1}|\alpha(w)|1\rangle $ where 
\[ \langle\mathbbm{1}| = \int_{P_\mass} \langle 1|g dg =\int_{i\ft+\mass}\langle 1|e^{2\pi \alpha(\hat{\mu}_{\FI}(\LE))}.\]
\begin{lemma}  The integral 
    $ \langle\mathbbm{1}|$ defines a tempered distribution.  
\end{lemma}
\begin{proof}
Consider any Schwartz function $f$.  We wish to show that the integral
\[\langle\mathbbm{1}|f\rangle =\int_{\LE\in \mass+i\ft} \Van(\LE)\langle\mathbbm{1}|e^{2\pi \alpha(\hat{\mu}_{\FI}(\LE))}|f\rangle d\LE \]
converges.  
We can now think of the eigenvalues $a_j(\LE)$ as functions of $\LE$, and for purposes of computing $\langle\mathbbm{1}|e^{2\pi \LE}|f\rangle$, we can divide $\mass+i\ft$ into cones by the sign of $a_j(\LE)$ for all $j$.  Having fixed one such cone, we can use partial Fourier transform on $\C^n$ to switch the signs of $a_j$ until all $a_j$ are positive.  Let $F$ be the transformed version of $f$;  note that this is always Schwartz, so $\|F\|_{\infty}<\infty$.  The inner product $\langle\mathbbm{1}|e^{2\pi \alpha(\hat{\mu}_{\FI}(\LE))}|f\rangle$ is bounded above by:
\begin{equation}
    \frac{2^n}{\pi^n}\int_{\C^n}e^{2\pi \FI(\mass)}\prod_{j=1}A_je^{-A_j^2|z_j|^2}F(z,\bar{z})dzd\bar{z}\leq \frac{2^n}{\pi^n}\|F\|_{\infty} \int_{\C^n} \prod_{j=1}^n A_je^{-A_j^2|z_j|^2}dzd\bar{z}=\frac{2^n\|F\|_{\infty}}{\prod_{j=1}^nA_j}.  
\end{equation}
Since there are only finitely many such transforms, there is a global constant $C'$ such that $C'\geq 2^n\|F\|_{\infty}$ for all such transforms.  This shows that for some constant $C$, we have \[|\Van(\LE)\langle\mathbbm{1}|e^{2\pi \alpha(\hat{\mu}_{\FI}(\LE))}|f\rangle|<Ce^{-\xi|\LE|}\] for all $\LE$ where as in the proof of \cref{th:integral-twisted-trace}, we let $\xi>0$ denote the minimum of the function \[m(\eta)=\sum_{j=1}^n |a_j(\eta)| - \sum_{\alpha\in \roots} |\alpha(\eta)|\] on the space of $\eta\in \ft$ of norm 1.  Thus, this integral is finite, and indeed, the distribution is tempered.  
\end{proof}

\begin{theorem}\label{th:reduced-positive}
    Assume that $(G,V)$ is conical, the moment map $\mu\colon V\to \mathfrak{g}^*_{\C}$ is flat, and $\mass=0,\FI=0$.  Then, the trace $\operatorname{Tr}_\mass^G$ is positive.  
\end{theorem}
\begin{proof}
First we need to consider the set $I=\{w\in \Weyl|\alpha(w)\mathbbm{1}=0\}$.  Of course, this is a left ideal of the Weyl algebra $\Weyl$.

We have $\alpha(\hat{\mu}_{\zeta}(\beta))\in I$, since this element exactly computes the Lie derivative along the orbits of the group that we have averaged over.   Thus, the classical limit of this ideal defines a subvariety $\mathbb{V}(I)$ of the variety $\mu^{-1}(0)$.  By the flatness of the moment map, the scheme $\mu^{-1}(0)$ is reduced, so either:
\begin{enumerate}
    \item $I$ is generated by $\alpha(\hat{\mu}_{\zeta}(\beta))$ for all $\beta\in \fg_{\C}$, or 
    \item $\dim \mathbb{V}(I)< n-\gamma$, where $\gamma=\dim G$.
\end{enumerate} 
Using flatness again, we find that the categorical quotient $\mu^{-1}(0)/\!\!/G$ has dimension $n-2\gamma$, and $\mathbb{V}(I)/\!\!/G=\mathbb{V}(I^G)$ has dimension $\dim \mathbb{V}(I)-g$, so we want to show that $\mathbb{V}(I^G)$ also has dimension $n-2g$.

Note that we can get any polynomial in $z_j$'s and $\bar{z}_j$'s times the Gaussian as $\alpha(w)|1\rangle$;  these functions are dense in $L^2(\C^n)$, so we can approximate any function $f\in \mathcal{S}(\C^n)$ arbitrarily well with $\alpha(w)|1\rangle$.  Furthermore, $\alpha(w)|1\rangle$ is $G$-invariant if and only if $w$ is $G$-invariant (or equivalently, $G_{\C}$-invariant), so if $f$ is $G$-invariant, we can assume that $w$ is $G_{\C}$-invariant.

In particular, if $h$ is any smooth bounded $G_{\C}$-invariant function on $\C^n$, then we can choose $w_j\in \Weyl^G$ such that $\alpha(w_j)|1\rangle$ converges in $\mathcal{S}(\C)$ to $h|1\rangle$. Of course, this means that $\alpha(w_j)|1\rangle$ converges as a tempered distribution to $h|\mathbbm{1}\rangle$. 

Since $G_{\C}$ is real $2\gamma$-dimensional, this is only possible if there are $n-2\gamma$ many algebraically independent $G$-invariant functions on $\mathbb{V}(I^G)$.  This shows that $\dim \mathbb{V}(I^G)\geq n-2\gamma$, so $I$ is generated by $\alpha(\hat{\mu}_{\zeta}(\beta))$ as a left ideal.

Thus, if $w\notin I$, then we have $\alpha(w)\mathbbm{1}\neq 0$.  
Note that for any $G$-invariant function $h$, we have 
\begin{equation}
    \langle h\mathbbm{1} | \alpha(w) 1\rangle=\int_G \langle {1} |\bar{h}g| \alpha(w) 1\rangle dg=\int_G \langle {1} | \bar{h}| \alpha(w) \mathbbm{1}\rangle dg=\langle {1} |\bar{h} f| \mathbbm{1}\rangle. 
\end{equation}
Thus, we have that \begin{equation}
\langle\alpha(w)\mathbbm{1}|\alpha(w)1\rangle=\langle f\mathbbm{1}|\alpha(w)1\rangle=\langle 1|\bar{f}| \alpha(w)\mathbbm{1}\rangle=\langle 1|\bar{f}f|\mathbbm{1}\rangle.
\end{equation}
This is the integral of a positive function on $G\times \C^n$, and thus is positive.  
\end{proof}
\section{Abelian theories}
\subsection{Analysis of abelian theories}

In this section, we assume that $G$ is abelian.  Since $G$ is connected and compact, we must have an isomorphism $G\cong U(1)^k$ for some $k$.  We can then choose an orthonormal Darboux basis of $V$ where $G$ acts diagonally.  
We can analyze this case with our starting point given by the exact sequence of tori 
\begin{align}
    1\rightarrow G \xrightarrow{i}D=U(1)^n\mathrel{\mathop{\rightleftarrows}^{\mathrm{p}}_{\mathrm{s}}}F\rightarrow 1,
\end{align}
where $D$ is the full subgroup of $USp(V,\Omega)$ which is diagonal in this basis.  That is, if we have coordinates on the hypermultiplet $(x_1,\dots, x_n;y_1,\dots, y_n)$, the action is given by the diagonal matrix $\operatorname{diag}(d_1,\dots, d_n,d_1^{-1},\dots, d_n^{-1})$.  

It is always possible to choose a group homomorphism $s: F\rightarrow D$ such that $p\circ s$ is the identity on $F$ (that is, splitting the exact sequence); nothing we will do depends on this choice, but it will be notationally convenient. 
We can pass the previous sequence to character (charge) lattices by considering $X^*(-)=\operatorname{Hom}(-,U(1))$:
\begin{align}
    0 \leftarrow X^*(G)\xleftarrow{i^*}\mathbb{Z}^n\mathrel{\mathop{\leftrightarrows}^{p^*}_{s^*}}X^*(F)\leftarrow 0.
\end{align}
If we choose isomorphisms $G\cong U(1)^k$ and $F\cong U(1)^d$, then we have induced isomorphisms $X^*(G)\cong \mathbb{Z}^k, X^*(F)\cong \mathbb{Z}^d$.  Thus, the maps $i^*,p^*$ and $s^{*}$ are given by $k\times n, n\times d$ and $d\times n$ matrices, respectively.  The property $i^* p^{*}=0$ and the relevant ranks show that the columns of $p^{*}$ are a basis of the nullspace of $i^{*}$, and similarly the rows of $i^{*}$ are a basis of the left nullspace of $p^{*}$.  On the other hand, $s^*$ is a left inverse of $p^*$.  
In the notation used in \cite{gaiottoSphereCorrelation2020}, the relevant matrices are $Q=(i^*)^T$ and $q=(s^*)^T$. This is essentially the same information that we will use later in this section to construct the hyperplane arrangements $i^*$ and $p^*$. It is perhaps more illustrative to look at a particular example. 

\begin{example}We'll use as a running example SQED$_n$, the theory with gauge group is $U(1)$ and matter $n$ copies of the usual representation on $\C^n$, so the sequence of characters is just 
\begin{align}
    0\leftarrow\mathbb{Z}\leftarrow\mathbb{Z}^n\leftarrow X^*(F)\leftarrow 0.
\end{align}
The matrix for the inclusion $G\hookrightarrow U(1)^n$ is given by 
\begin{align}
    i^*=[1,\dots,1].
\end{align}
On the other hand, the map $p^*$ is an isomorphism onto the perpendicular space $\mathfrak{f}=\mathfrak{g}^\perp$ as a vector subspace in $\mathbb{R}^n$. With this picture in mind, we can identify the columns of $p^*$ as a basis of the null space, and construct $s^*$ as a left inverse:
\begin{align}
    p^*=\begin{pmatrix}
    -1 & \cdots & -1 \\
    1 & \cdots & 0 \\
   \vdots &  \ddots & \vdots \\
    0 & \cdots & 1
    \end{pmatrix}\qquad    s^*=\begin{pmatrix}
   0& 1 & \cdots & 0 \\
\vdots & \vdots   &  \ddots &\vdots  \\
  0&  0 & \cdots  & 1
    \end{pmatrix}. 
\end{align}
We can encode this now in the language of gauge and flavour charges by introducing $\mathscr{Q}=(Q, q)$, given just by 
\begin{align}
    \mathscr{Q}=\begin{pmatrix}
    1 & 0 & \dots & 0 \\
    1 & 1 & \cdots & 0\\
    \vdots & \vdots   &  \ddots &\vdots  \\
    1 & 0 & \cdots  & 1
    \end{pmatrix}.
\end{align}
\end{example}

The Hamiltonian reduction $\AHiggs$
of $\Weyl$ by $G$, that is, the chiral ring of the Higgs branch, has been studied in the math literature as the hypertoric enveloping algebra \cite{MVdB,BLPWtorico} and described from a physics perspective in \cite[\S 3.3]{bullimoreCoulombBranch2017}.

We can most easily visualize these data by drawing hyperplane arrangements based on them: the map $p_*=(p^*)^T$ gives $n$ elements of $\mathfrak{f}$, which we denote by $h_1,\dots, h_n$. These are the images in $\mathfrak{f}_{\C}$ of the basis $\dh_j$ of $\mathfrak{d}_{\C}$ such that $\hat{\mu}(\dh_j)=x_j\partial_j+\frac{1}{2}$. Note that the coefficients of the mass parameter $\mass$ in terms of this basis will be {\it imaginary}.

Hence, to each abelian gauge theory, we associate an infinite family of hyperplanes:
\begin{align}
        \{h_j\in\mathbb{Z}\}\subset \mathfrak{f}^\vee_{\C}. 
    \end{align}
An important role will be played by the linear dependence relations between the elements $h_j$, which can be expressed as a matroid structure.  For simplicity, we assume that none of the elements of $h_j$ are 0, and that for each $j$, $h_j$ can be written as a linear combination of the other $h_j$'s (that is, the corresponding matroid has no loops or coloops).  

The functions $h_j$ determine coordinates on this space, which are not linearly independent. We may interpret these coordinates as the signed distance from the hyperplanes $H_j=\{\mathbf{x}\mid h_j(\mathbf{x})=0\}$.

We identify each point with an element of the Weyl algebra:
\begin{align}   
    m^{\Vec{v}} =\prod_{j=1}^n m_j^{v_j}\qquad m_j^{v_j}= \begin{cases}
    &x_j^{v_j}, \quad \textrm{if} \quad v_j\geq 0, \\
    &\partial_j^{-v_j}, \quad \textrm{if} \quad v_j<0
    \end{cases}.
\end{align}
The linear dependence between coordinates (that is, the fact that we consider a point in $\mathfrak{f}^{\vee}_{\C}$) means that this element is invariant under the group $G$, and these elements generate the invariants $\Weyl^G$ as a free module over the polynomial ring $\C[\dh_j]=U(\mathfrak{d}_{\C})$.

Our construction will depend on two parameters: 
\begin{itemize}
    \item Mass parameters, $\mass\in \mathfrak{f}$.
    \item FI parameters, $\FI\in\mathfrak{g}^\vee$.
\end{itemize}

We can naturally identify $h_j$ with the image of $x_j\partial_j+\frac{1}{2}\in \Weyl^G$ in the quotient $\AHiggs$.  When performing the quantum Hamiltonian reduction, we mod out by inhomogeneous linear relations on these elements, which set $X=\langle X,\FI\rangle$ for $X\in \mathfrak{g}$.  The vanishing set of the functions $X-\langle X,\FI\rangle$ is precisely the coset $\FI+\mathfrak{f}^{\vee}_{\C}$, so the image of $U(\mathfrak{d}_{\C})$ in $\AHiggs$ is naturally the same as the polynomial functions on this coset.  

If we fix a choice of $\Tilde{\zeta}\in \FI+\mathfrak{f}^{\vee}$, then we can use this to identify $\FI+\mathfrak{f}^{\vee}$ with $\mathfrak{f}^{\vee}$.  The image of the hyperplanes $h_j=0$ in $\FI+\mathfrak{f}^{\vee}$ will be displaced versions of these hyperplanes in $\mathfrak{f}^{\vee}$ defined by 
\begin{align}
    H_j=\{h_j+\Tilde{\zeta}_j=0\}.
\end{align}

\subsection{Twisted traces of modules}

The article \cite{BLPWtorico} studies a category of modules over $\AHiggs$  called category $\cO$, which depends on the mass parameter $\mass$.
\begin{defi}
	A finitely generated module $M$ over $\AHiggs$ lies in category $\cO$ if the element $2\pi \hat{\mu}_{\FI}$ acts locally finitely, with each generalized eigenspace finite-dimensional and there is an upper bound on the real parts of its spectrum.  
\end{defi}

In more concrete terms, this means that $2\pi \hat{\mu}_{\FI}$ acts on some basis by an upper-triangular matrix in Jordan canonical form with all Jordan blocks of finite length and finitely many Jordan blocks for each eigenvalue, together with the upper bound on the spectrum.

There are special modules in this category that we call {\it Verma modules}.  These should be viewed as generalizations of the vacuum module over the Weyl algebra.  We'll discuss the combinatorics and structure of these modules below.  Note that some related calculations are made in \cite[Appendix B]{bullimoreBoundariesVermas2021}.  There is one for each subset of $I\subset [1,n]$ such that $\{p_*(h_i)\}_{i\in I}$ is a basis of $\mathfrak{f}$.  Note that this set must have size $|I|=d$.  
For each basis, there is a unique expression for the mass $\mass\in \mathfrak{f}$ in terms of $h_i$: 
\begin{equation}
	\mass=\sum_{i\in I}z_ih_i.
\end{equation}
The sign $\operatorname{sign}(z_j)=\signs_I(j)$ defines a function $\signs_I\colon I\to \{\pm 1\}$ for each basis $I$.

Let $\QI$ be the minor of the matrix $p^*$ given by the rows corresponding to $I$; the basis condition is exactly that this minor has determinant $\pm 1$. It will be useful to consider the product $R_I=Q(\QI)^{-1}$; this is a matrix whose rows corresponding to $I$ are the different unit vectors of $\mathbb{R}^k$.  Unlike $Q$, this matrix does not depend on any choice other than that of the set $I$.  The other rows of $\qI$ are the coefficients for the expansion in $\mathfrak{f}_{\mathbb C}$: 
\[p_*(h_s)=\sum_{t=1}^{d}\qi_{st}p_*(h_{t_I})\]
where $t_I$ is defined so that $1_I<2_I<\cdots <d_I$ are the elements of $I$ (that is, the $t$th row of the minor $\QI$ is the $t_I$th row of $Q$).
More useful for us will be the matrix $R_I'=-R_I\cdot \operatorname{diag}(\signs_I(1_I),\dots, \signs_I(d_I))$ defined by multiplying the columns of $-R_I$ by the corresponding sign.

\begin{defi}
	The Verma module $\Verma_I$ is the unique module over $\AHiggs$ such that:
	\begin{enumerate}
		\item The module $\Verma_I$ is generated by a single vector.
		\item The module $\Verma_I$ has a basis consisting of simultaneous eigenvectors for $\hat{\mu}(h_i)$ with eigenvalue $\tilde{c}_i=-\signs_I(i)(c_i+\frac{1}{2})$ for each $(c_i)\in \Z_{\geq 0}^I $, with the generating vector corresponding to $c_i=0$ for all $i$.
	\end{enumerate}
\end{defi} 
For a fixed value of $\FI$, the different possible modules in category $\cO$ are described by \cite[Th. 4.9]{BLPWtorico}.  In particular, if $\FI$ is generic (in the sense of avoiding a countable union of hyperplanes), then each of the hyperplane arrangements $\mathcal{H}'$ appearing in this theorem consist of $n$ hyperplanes meeting like the coordinate axes; in this case, the corresponding block of category $\cO$ consists only of direct sums of copies of a single Verma module.  One can see this directly from the fact that if there is a homomorphism between Verma modules $\Verma_I\to \Verma_J$, then for all $i\in I\cup J$, the operator $\hat{\mu}(h_i)$ must have all eigenvalues in $\Z+\frac 12$ on both $\Verma_I$ and $\Verma_J$;  if $j\notin I$, then $\hat{\mu}(h_j)$ acts on $\Verma_I$ with eigenvalues in  $\Z+\frac 12$ if $\FI$ lies in a proper countable union of hyperplanes.
In particular, we find that:   
\begin{lemma}\label{lem:generic-verma}
	If $\FI$ is generic, then all Verma modules in category $\cO$ are simple and every module in category $\cO$ is a direct sum of Verma modules.  
\end{lemma}

The theorems of \cite[\S 4.2]{etingofShortStarProducts2020} imply that the twisted trace of this module is well-defined for appropriate generic masses and depends rationally on $e^{2\pi \mass}$, but we can obtain more precise results by exploring the structure of this module.  Thus, consider the vectors $\vec{r}_j$ given by the $j^I$th column of the matrix $R_I'$.  This is chosen so that the element $m^{\vec{r}_j}$ is $G$ invariant, and for $j,k\in I$, we have that
\begin{equation}
	[\hat{\mu}(h_{j}),m^{\vec{r}_k}]=-\signs_I(k)\delta_{jk}m^{\vec{r}_k}.  
\end{equation} 

The Verma module is generated by any non-zero vector $v$ of eigenvalue $\frac{1}{2}\signs_I(j)$ for the operator $\hat{\mu}(h_j)$, and the other eigenvectors are constructed as 
\begin{equation}
	v_{\vec{c}}=m^{\sum_{j\in I} c_j\vec{r}_j}v=\prod_{j\in I} (m^{\vec{r}_j})^{c_j}v
\end{equation} 
These form a basis of the Verma module as $\vec{c}$ ranges over  $ \Z_{\geq 0}^I $.

Let $\colsum_j(I)$ be the sum of the entries of $\vec{r}_j$ for $j\in I$. Note that $\onef$ acts on $m^{\vec{r}_j}$ by $(-1)^{\colsum_j(I)}$, and so we can extend the action of $\onef$ to $\Verma_I$ by having it act on $v$ by $\onef v=v$, so 
\begin{equation}\label{eq:onef}
	\onef v_{\vec{c}}=(-1)^{\sum_{j\in I}\colsum_j(I)c_j}v_{\vec{c}}.
\end{equation}

\begin{example}
	For $SQED_3$, we can choose 
	\[\mass = 2h_1+h_2=-h_2-2h_3=h_1-h_3\]
	in which case we have that
	\[\signs_{\{1,2\}}(1)=\signs_{\{1,2\}}(2)=\signs_{\{1,3\}}(1)=1\qquad \signs_{\{2,3\}}(2)=\signs_{\{2,3\}}(3)=\signs_{\{1,3\}}(3)=-1\]
	We have a single FI-parameter $\zeta$, which gives the deformed equation $h_1+h_2+h_3=\zeta$, and the highest weights of the Verma modules are:
	\begin{equation}
    \lambda_{\{1,2\}}=(-\frac{1}{2}, -\frac{1}{2}, \zeta+1) \qquad   \lambda_{\{1,3\}}=(-\frac{1}{2}, \zeta, \frac{1}{2}) \qquad \lambda_{\{2,3\}}=(\zeta-1, \frac{1}{2}, \frac{1}{2}).
    \end{equation} 
    The matrix $p_*$ is given by 
    \begin{equation}
    	 p_*=\begin{bmatrix}
    	-1 & -1 \\
    	1 &0\\
    	0 & 1
    \end{bmatrix}
    \end{equation}
   so the matrices $R_I$ are given by 
   \begin{equation}
   	R_{\{1,2\}}=\begin{bmatrix}
    	1 &0\\
    	0 & 1\\
  	-1 & -1 
    \end{bmatrix} \qquad R_{\{1,3\}}=\begin{bmatrix}
    	1 &0\\
		-1 & -1 \\
    	0 & 1
    \end{bmatrix}
    \qquad 	R_{\{2,3\}}=
    \begin{bmatrix}
       	-1 & -1 \\
    	1 &0\\
    	0 & 1
    \end{bmatrix}
   \end{equation}
   For example, if $I=\{1,3\}$, then $\vec{r}_1=(-1,1,0)$ and $\vec{r}_3=(0,-1,1)$.  
\end{example}

Note, this shows that the trace  $\Tr_{\Verma_I}(w)=\operatorname{Tr}(M;	\onef \exp(2\pi \hat{\mu}_{\FI}) w)$ is easily computable.  Of course, we can reduce to the case were $w={h}_{1_I}^{u_{1_I}}\cdots {h}_{d_I}^{u_{d_I}}$ is a monomial in $h_i$ for $i\in I$.  Let 
\begin{equation}
	\sorc_{\colsum}^{(u)}(z)=\begin{cases}
		\sech^{(u)}(z) & \colsum\text{ odd}\\
		\csch^{(u)}(z) & \colsum\text{ even}\end{cases}
\end{equation}
\begin{lemma}\label{lem:Verma-trace}
$\displaystyle \Tr_{\Verma_I}({h}_{1_I}^{u_{1_I}}\cdots {h}_{d_I}^{u_{d_I}})=\prod_{j\in I}\frac{(\signs_I(j))^{u_j+\colsum_j(I)+1}}{2^{u_j+1}}\sorc^{(u_j)}_{\colsum_j(I)}(\pi z_j) $
\end{lemma}
\begin{proof}
	We can compute this trace by finding the eigenvalue on $v_{\vec{u}}$.  Since this is an eigenvector with eigenvalue $\tilde{c}_j=-\signs_I(j)(c_j+\frac{1}{2})$ for $h_j$, we have that $f(h_{1_I},\dots, h_{d_I})$ acts by $f(\tilde{c}_{1_I},\dots,\tilde{c}_{d_I})$.  The automorphism $\onef$ acts as in \eqref{eq:onef}.  Finally, $\exp(2\pi \hat{\mu}_{\FI})$ acts by $\exp(2\pi \sum_{j\in I}z_j\tilde{c}_{j}). $ This shows that:
\begin{equation}
	\Tr_{\Verma_I}(f(h_{1_I},\dots, h_{d_I}))
	=\sum_{\mathbf{u}\in \Z_{\geq 0}^I} (-1)^{\sum_{j\in I} \colsum_j(I)c_j}f(\tilde{c}_{1_I},\dots,\tilde{c}_{d_I})\exp(2\pi \sum_{j\in I}z_j\tilde{c}_{j}) 	
\end{equation}
	By construction, $\tilde{c}_{j}$ has sign $-\signs_I(j)=-\operatorname{sign}(2\pi z_j)$, so this last term is exponentially decreasing as $c_j$ increases, showing that the sum converges by the ratio test.

	We can separate variables and expand this sum as a product:
\begin{equation}
	\Tr_{\Verma_I}({h}_{1_I}^{u_{1_I}}\cdots {h}_{d_I}^{u_{d_I}})= \prod_{j\in I} \sum_{c_j\in \Z_{\geq 0}}(-1)^{\colsum_j(I)c_j}\tilde{c}_{j}^{u_j}\exp(2\pi z_j\tilde{c}_{j})
\end{equation}
	We can then simplify the terms of this product by using the series expansions (\ref{eq:sec-sum1}--\ref{eq:sec-sum4}) with $x=-\pi z_j\signs_I(j)<0$:
	\begin{align}\label{eq:Verma-telescope}
	\sum_{c_j\in \Z_{\geq 0}}(-1)^{c_j\colsum_j(I)}\tilde{c}_{j}^{u_j}\exp(2\pi z_j\tilde{c}_{j})&= \Big(\frac{\signs_I(j)}{2}\Big)^{u_i}(e^{x} +(-1)^{\colsum_j(I)} 3^{u_i}e^{3x}+5^{u_i}e^{5x}+\cdots )\\
		&=\frac{(\signs_I(j))^{u_j}}{2^{u_j+1}}\sorc^{(u_j)}_{\colsum_j(I)}(\pi z_j\signs_I(j))\\
        &= \frac{(\signs_I(j))^{u_j+\colsum_j(I)+1}}{2^{u_j+1}}\sorc^{(u_j)}_{\colsum_j(I)}(\pi z_j)
		\end{align}
		The result follows.
	\end{proof}

\begin{prop} For any module $M$ in category $\cO$, the trace $\Tr_M(w)=\operatorname{Tr}(M;	2\pi \hat{\mu}_{\FI} w)$ is convergent and a twisted trace for the automorphism $g^{-1}$.  
\end{prop}
\begin{proof}
Of course, if this holds for a module $M$, then it holds for any submodule of $N\subset M$ and the quotient $M/N$.  Similarly if it holds for $N$ and $M/N$, then $\Tr_{M}=\Tr_{M/N}+\Tr_{N}$, and so it holds for $M$.  

Every $M$ is a quotient of a projective module, and every projective has a Verma filtration.  Thus, it is only necessary to prove this for a Verma module, which follows by Lemma \ref{lem:Verma-trace}.
\end{proof}

\subsection{Sphere trace}
For general abelian theory, we have the exact series of the groups and the characters 
\begin{equation}
\begin{tikzcd}
     1\ar{r} &G=U(1)^k\ar{r}{i}&D=U(1)^n\ar[bend left=10]{r}{s}&F=U(1)^{N-r}\ar[bend left=10]{l}{p}\ar{r}&1,\\
         0&\mathbb{Z}^k \ar{l}&\mathbb{Z}^n\ar{l}{Q^T}\ar[bend left=10]{r}{q^T}&\mathbb{Z}^{d}\ar[bend left=10]{l}{p^*}&0\ar{l}.
\end{tikzcd}
\end{equation}
As before, we have an  invertible matrix $\mathscr{Q}$ whose entries are integral as are those of its inverse, which has  the form
\begin{equation}
    \mathscr{Q}=(Q,q) \qquad Q=(Q_A^i), q=(q_A^a).
\end{equation}

Since $\mathbb{V}=1$ in this case, we have that the integral \eqref{eq:Higgs-trace} is given by:
\begin{equation}\label{eq:Higgs-trace-ab}
    \operatorname{Tr}_\mass^G(h_1^{u_1}\cdots h_n^{u_n})=\int_{i\mathfrak{{g}}+\mass}  \operatorname{Tr}_{\LE}(h_1^{u_1}\cdots h_n^{u_n})=\int_{i\mathfrak{{g}}+\mass} e^{2\pi  \zeta(\LE)}\prod_{j=1}^n\frac{1}{2^{u_j+1}}\sech^{(u_j)}(\pi  h_j^{\vee}(\LE))
\end{equation}
Recall that the functions $h_i^{\vee}$ give the eigenvalues of the action of an element of the Lie algebra on $\C^n$ and thus are real on $i\fg$ and imaginary on $\fg$. Thus, in this basis, $\mass$ is real and $\FI$ is imaginary.  
We can consider this as the integral of a distribution on $\mathfrak{d}$, supported on the subspace $i\mathfrak{{g}}+\mass$.  

Recall that using contour integration and the sum of residues, we can calculate the Fourier transform
\begin{equation}
\int_{\mathbb R} dx e^{2 \pi i xy}\sech(\pi x)=\sum_{c=0}^{\infty}2(e^{-\pi y/2} - e^{-3\pi y/2}+\cdots )=\sech(\pi y). 
\end{equation}

\begin{theorem}
    $\operatorname{Tr}_\mass^G(h_1^{u_1}\cdots h_n^{u_n}))=\int_{i\mathfrak{f}^{\vee}-i\FI}e^{2\pi i \mass}\prod_{j=1}^n\frac{1}{2}(\pi i h_j(\alpha))^{u_j}\sech(\pi  h_j(\alpha))$
\end{theorem}
\begin{proof}
We can write the integrand as the product $e^{2\pi  \zeta(\LE)}\prod_{j=1}^n\frac{1}{2^{u_j+1}}\sech^{(u_j)}(\pi  h_j^{\vee}(\LE))$.  Considering this as a function $f$ on $i\mathfrak{d}\cong \mathbb R^n$, we can consider the integral \eqref{eq:Higgs-trace-ab} as the value of the pushfoward of this function from $i\mathfrak{d}$ to $i\mathfrak{f}$ at the point $\mass\in i\mathfrak{f}$.   
We can evaluate this projection a second way using the Fourier transform $\hat f\colon i\mathfrak{d}^*\to \mathbb C$, which is given by 
\begin{equation*}
\hat{f}=\prod_{j=1}^n\frac{1}{2}(\pi i h_j(\alpha+i\FI))^{u_j}\sech(\pi  h_j(\alpha+i\FI))
\end{equation*}
by the rule that $\widehat{\frac{\partial g}{\partial h_i^{\vee}}}=2\pi i h_i^{\vee}\hat{g}$, and that Fourier transform commutes with product of uncorrelated functions.  
By the slice-projection theorem, the pushforward of $f$ to $i\mathfrak{f}$ is the same as the inverse Fourier transform of the restriction of $\hat{f}$ to the subspace $i\mathfrak{f}^{\vee}$.  Thus, we can
we can rewrite $\operatorname{Tr}_\mass^G(h_1^{u_1}\cdots h_n^{u_n}))$ as
\begin{equation*}
\int_{i\mathfrak{f}^{\vee}-i\FI}e^{2\pi i \mass}\hat{f}(\alpha-i\FI)=\int_{i\mathfrak{f}^{\vee}-i\FI}e^{2\pi i \mass}\prod_{j=1}^n\frac{1}{2}(\pi i h_j(\alpha))^{u_j}\sech(\pi  h_j(\alpha)),
\end{equation*}
as desired.
\end{proof}

We can interpret this formula as a ``dual'' construction of our trace, which comes from considering the algebra $\AHiggs$ as the quantum Coulomb branch of the theory with gauge group $\check{F}$, the Langlands dual of $F$, acting on $\mathbb C^n$.

Now, we turn to considering the relationship between these traces and the traces coming from the action on Verma modules.  
In our argument below, we will want to think about the geometry of the complex affine space $\fg_{\C}+\mass$.  We can consider the functions $h_i^{\vee}$ restricted to the affine space, but these are no longer independent coordinate functions.  Basic linear algebra shows that:
\begin{itemize}
	\item The equations $h_j^{\vee}=c_j$ for $j\in J\subset [1,n]$ have a unique common solution for all $c_j$ if and only if the vectors $\{h_j^{\vee}\mid j\}$ is a basis of $\mathfrak{g}^*$.  
\end{itemize}
Thus, for generic mass, the maximal number of $h_i^{\vee}$ having a common zero is $r$, and there is one such common zero for every basis $I=[1,n]-J$.  Below, we will always use $I$ and $J$ below to denote complementary subsets of this form.  Furthermore, since $h_i^{\vee}(\mass)\in \mathbb{R}$, all of these solutions lie on $i\mathfrak{{g}}+\mass$.  Let $\mathbf{x}_I$ denote the unique common solution to $h_j^{\vee}=0$ for $j\notin I$; the values of $h_i^{\vee}$ at this point are determined by $\mass$ and are necessarily real.  

Each of these bases cuts the space $i\mathfrak{{g}}+\mass$ into $2^r$ orthants, and there is exactly one of these orthants where $i\zeta(\beta)$ is bounded below.  Algebraically, we can find this by taking the unique expansion $\FI|_{i\fg}=\sum_{j\in J}iz^!_jh_j^{\vee}$.  Let $\signs_J^!(j)=\operatorname{sign}(z_j^!)$.  Similarly, if $Q^!$ is the dual charge matrix given by the map $i^*$, then we can define $R^!_J$ analogously to $R_I$, and these matrices will be related by $R^!_J=-R_I^T$.  Let $\rowsum_j(J)$ be the analogue of $\colsum_j(I)$ for the dual theory, i.e. the sum of entries in the column of $R^!_J$ corresponding to $j$ (which is a row of $-R_I$) times $\signs_J^!(j)$.

Note that if $2\cosh\pi h_i^\vee(\mathbf{x})=0$, then we must have that $h_i^{\vee}$ has real part 0, since all zeros of $\cosh$ are on the imaginary axis.  Thus, all the common zeroes of $  2\cosh\pi h_j^\vee(\mathbf{x})$ are the common solutions to the equations, for each fixed $c_j\in \Z, j\in J$:
\begin{equation}\label{eq:zeroes-1}
	h_j^\vee(\mathbf{x})=i(c_j+\frac{1}{2}).
\end{equation}
These will always lie in $\mathbf{x}_I+\fg$.  

Before doing any more technical discussion, let us just point out that this suggests an approach to calculating the integral \eqref{eq:Higgs-trace-ab}: by deforming our integration contour, we can write it as a sum of the residues at all zeros ``on one side'' of the integration contour $i\fg +\mass$.  We have to be careful about what ``on one side'' means, but we will elucidate this below.  These zeros naturally fit into families depending on which factors in the denominator vanish, and the residue at each point in this family is a scalar multiple of the trace on a Verma module.  Furthermore, this formula is a hint of the 3-d mirror symmetry perspective on this trace: the poles of the integrand are in bijection with the basis vectors of a Verma module for the Coulomb branch of this theory (which is also a Higgs branch, but for a dual theory), and in fact, up to sign, the coefficients of the expansion of the sphere trace into terms of the Verma modules is the trace of the identity on the Verma modules for the Coulomb branch.  It would be interesting to explain this phenomenon more systematically.

We can view $\Tr^G_{\mass}$ as an iterated contour integral and thus calculate this integral as a sum over residues.  More precisely: given a vector $\LE\in \ft$, we can consider the cycle $i\ft+\mathbb{R}_{\geq 0}\LE$, which has boundary $i\ft$.

	Let us compute the integral of \eqref{eq:Higgs-trace-ab}.  Fix a small real number $\epsilon>0$, and for each basis $J$, let 
	\[ U_{J}=\{\mathbf{x}\in \fg_{\C}+\mass \mid \Im \big(\signs_J^!(j)h_j^\vee(\mathbf{x})\big)\geq 0\text{ and } |\Re h_j^\vee(\mathbf{x})|\leq \epsilon\}. \]
	Each of these regions is a product of $k$ copies of a segment of length $2\epsilon$ in the direction of $i\fg$, and a ray in the direction of $i\fg$.  We can deform the subspace $i\fg+\mass$ in the direction of $\FI$ between these regions, assuming that $\epsilon$ is small enough that none of these regions overlap.  This will cause this contour to limit onto the union of the boundaries of $U_J$, so by Stokes' theorem, the integral over $i\fg+\mass$ is the same as the sum of the integrals over this boundary.

Furthermore, since the integral of an exact form on a chain with trivial boundary is 0, we have that $\operatorname{Tr}_J=\int_{\partial U_J} \Van(\LE) \operatorname{Tr}_{\LE}(w)$ is also a trace on the Hamiltonian reduction $\AHiggs$ such that:
\begin{equation}
	\operatorname{Tr}^{G}_{\mass}(w)=\sum_{J}\operatorname{Tr}_J(w)
\end{equation}
As above, we let $I=J^c$ be the complement of $J$.  Let \begin{equation}
\weird(j)=\begin{cases}
1 & j\equiv 0,1 \pmod 4\\
-1 & j\equiv 2,3 \pmod 4
\end{cases} 
  \end{equation}
  This slightly odd function gives us the useful formula:
\begin{equation}
\sech^{(u)}\big(x+\frac{\pi i(j-1)}{2}\big)=  \weird(j)\sorc^{(u)}_j=\begin{cases}
\csch^{(u)}(x) & j\equiv 0 \pmod 4\\
\sech^{(u)}(x) & j\equiv 1 \pmod 4\\
-\csch^{(u)}(x) & j\equiv 2 \pmod 4\\
 -\sech^{(u)}(x)& j\equiv 3 \pmod 4.
\end{cases}  	
  \end{equation}
\begin{theorem}
$	\displaystyle \Tr_J=\Bigg(\det(Q^!_J)\prod_{j\in I}\weird(\tilde{\colsum}_r(J))\signs_I(j)^{\colsum_j(I)+1}\prod_{j\in J}\signs^!_J(j)^{\colsum^!_j(J)} \sorc_{\rowsum_j(J)}(\pi z_j^!)\Bigg) \Tr_{\Verma_I}$
	\end{theorem}
\begin{proof}
To compute the integral over $\partial U_J$, we break the integral around this boundary into the sum of the integrals around the boundary of the set
 \begin{equation}
     \{\mathbf{x}\in \fg_{\C}+\mass \mid c_i+1\geq \Im \big(\signs_J^!(j)h_j^\vee(\mathbf{x})\big)\geq c_i\text{ and } |\Re h_j^\vee(\mathbf{x})|\leq \epsilon\}.
 \end{equation}
	This integral can then be factored into contour integrals in the variables $h_j^{\vee}$ for $j\in J$.
The result is the residue of the integrand at the unique point in this box where the function fails to be holomorphic, given by the unique solution to:
 \begin{equation}\label{eq:zeroes-2}
	h_j^\vee(\mathbf{x})=i\tilde{c}_j=i\signs_J^!(j)(c_j+\frac{1}{2}) \qquad j\in J, c_j\in \Z_{\geq 0}.
\end{equation}
 Note that since we only used $h_{i}$ for $i\in I$ in our product, at this point, only factors of $\sech(\pi  h_j^{\vee}(\LE))$ with $u_j=0$ have poles at this point.  Thus, we can calculate just on the basis of the fact that the residue of $\sech(x)$ at $x=2\pi i(c+\frac{1}{2})$ is $(-1)^{c+1}$.  
	
If we change coordinates to use $h_{1_J}^{\vee},\dots, h_{(n-k)_J}^{\vee}$,   then we have to account for the fact that $dz_1\wedge\cdots \wedge dz_k
=\det(Q^!_J)dh_{1_J}\wedge \cdots \wedge dh_{k_J}$. 	Thus, for the product $e^{2\pi \zeta(\LE)}\prod_{j=1}^n\frac{1}{2^{u_j}}\sech^{(u_j)}(\pi h_j^{\vee}(\LE))$ at the point given by \eqref{eq:zeroes-2}, this residue is
		\begin{equation}\label{eq:residue1}
			\det(Q^!_J)\prod_{i\in I}\frac{1}{2^{u_i+1}}\sech^{(u_i)}(\pi h_i^{\vee}(\LE))\prod_{j\in J}e^{-2\pi z_j^!\tilde{c}_j}(-1)^{c_j+1}\signs_J^!(j)
		\end{equation}  
        First note that for $r\in I$, we have an expansion: 
        \begin{equation}
        	h^{\vee}_r(\LE)=z_r+\sum_{s=1}^{n-k} (R^!_J)_{r^Is}h^{\vee}_{s_J}(\LE)=z_r+i \sum_{s=1}^{n-k} (R^!_J)_{r^Is}\tilde{c}_{s^J}        
        	\end{equation}
        	We can rewrite \begin{equation}
        		\sum_{s=1}^{n-k} (R^!_J)_{r^Is}\tilde{c}_{s^J} =\frac{1}{2}(\tilde{\colsum}_r(I)-1)+\sum_{s=1}^{n-k} (R^!_J)_{r^Is} \signs_J^!(s)c_{s^J}
        	\end{equation} 
 where we use the lift $ \tilde{\colsum}_r(I)=1+\sum_{s=1}^{n-k} (R^!_J)_{r^Is}\signs_J^!(j)$ of $\colsum_r(I)$ to an integer.  
        Note that all entries of $R^!_J$ are integers, and ${c}_{s^J}\in \Z$, so $h^{\vee}_i(\LE)\in z_i-\frac{i\colsum_r}{2}+i\Z$.  Also note the residue is unchanged by shifting this value by $2i\Z$, by the periodicity of $\cos$.  Thus, we find that
\begin{align}
        	\sech^{(u_r)}(\pi h_r^{\vee}(\LE))&=\sech^{(u_r)}(\pi(z_r+ i \sum_{s=1}^{n-k} (R^!_J)_{rs}\tilde{c}_{s^J}))\\
        	&=\sech^{(u_r)}(\pi(z_r+ \frac{i}{2}(\tilde{\colsum}_r(I) -1)+ i \sum_{s=1}^{n-k} (R^!_J)_{rs}\signs_J^!(s)c_{s^J}))\\
        	&= (-1)^{\sum_{s=1}^{n-k} (R^!_J)_{rs}\signs_J^!(s)c_{s_J}}\weird(\tilde{\colsum}_r(J))\sorc^{(u_r)}_{\colsum_r}(\pi z_r)
        \end{align} By definition $\sum_{r\in I} (R^!_J)_{r^Is}\signs_J^!(s)=\rowsum_s$, so this simplies the formula 
        \eqref{eq:residue1} to the product:
        \begin{equation}
        	\det(Q^!_J)\prod_{r\notin J}\frac{1}{2^{u_r+1}}\weird(\tilde{\colsum}_r(J)) \sorc^{(u_r)}_{\colsum_r(J)}(\pi z_i)
\prod_{j\in J}e^{-2\pi z_j^!\tilde{c_j}}(-1)^{\rowsum_j(J)c_j}\signs_J^!(j)
 \end{equation}
Summing over residues using the same argument as  \eqref{eq:Verma-telescope}, we find that 
\begin{equation}\label{eq:final-formula}
\Tr_J(h_{1_I}^{u_{1_I}}\cdots h_{d_I}^{u_{d_I}})=\det(Q^!_J)\prod_{r\notin J}\frac{1}{2^{u_r+1}}\weird(\tilde{\colsum}_r(J)) \sorc^{(u_r)}_{\colsum_r(J)}(\pi z_r) \prod_{j\in J}\signs^!_J(j)^{\rowsum_j(J)} \sorc_{\rowsum_j(J)}(\pi z_j^!)
\end{equation}
Applying \cref{lem:Verma-trace} completes the proof.
\end{proof}

\begin{cor}
	$\displaystyle\Tr^G_{\mass}=\sum_{I}\Bigg(\det(Q^!_J)\prod_{j\in I}\weird(\tilde{\colsum}_r(J))\signs_I(j)^{\colsum_j(I)+1}\prod_{j\in J}\signs^!_J(j)^{\colsum^!_j(J)} \sorc_{\rowsum_j(J)}(\pi z_j^!)\Bigg) \Tr_{\Verma_I}$.
\end{cor}

This confirms \cite[(1.5)]{gaiottoSphereCorrelation2020}.  Note that \eqref{eq:final-formula} is considerably more symmetric under duality, especially when applied to the identity, in which case it reads \begin{equation}\label{eq:final-formula2}
\Tr_{\mass}^G(1)=\sum_{I}\det(Q^!_J)\frac{1}{2^k}\prod_{r\notin J}\weird(\tilde{\colsum}_r(J)) \sorc_{\colsum_r(J)}(\pi z_r) \prod_{j\in J}\signs^!_J(j)^{\rowsum_j(J)} \sorc_{\rowsum_j(J)}(\pi z_j^!)
\end{equation}
This confirms \cite[(1.3)]{gaiottoSphereCorrelation2020}.

\subsection*{Acknowledgements}

D. G. and B. W. are supported by NSERC Discovery Grants.  This research was supported in part by Perimeter Institute for Theoretical Physics. Research at Perimeter Institute is supported by the Government of Canada through the Department of Innovation, Science and Economic Development Canada and by the Province of Ontario through the Ministry of Research, Innovation and Science.

\subsection*{Data availablity}

No datasets were generated or analyzed during the current study.

\subsection*{Competing Interests}

The authors have no competing interests to declare that are relevant to the content of this article.

\bibliographystyle{timesjhepbib}
\bibliography{bibliography.bib}

\end{document}